%% file: main.tex
\documentclass[letterpaper,final,twocolumn]{cls/IEEEtran}

\renewcommand{\baselinestretch}{.98}

\usepackage[T1]{fontenc}
\usepackage{amsmath,amssymb,amsfonts,amsthm,amscd,dsfont} 
\usepackage{graphicx}
\usepackage{color}
\usepackage{tikz}
\usetikzlibrary{positioning}
\usepackage[vlined,ruled,linesnumbered]{algorithm2e}

\usepackage[colorlinks]{hyperref}

\theoremstyle{plain}

\newtheorem{lemma}{Lemma}

\newtheorem{Proposition}{Proposition}

\theoremstyle{definition}
\newtheorem{defn}{Definition}
\newtheorem{ansatz}{Assumption}
\newtheorem{problem}{Problem}

\theoremstyle{remark}
\newtheorem{remark}{Remark}

\usepackage[compress,noadjust]{cite}
\usepackage{soul}

\newcommand{\oprocendsymbol}{\hbox{$\bullet$}}
\newcommand{\oprocend}{\relax\ifmmode\else\unskip\hfill\fi\oprocendsymbol}

\newcommand{\blue}[1]{\color{blue}{#1}}

\newcommand{\ubfu}{\underline{\bfu}}
\newcommand{\real}{\mathbb{R}}

\DeclareMathOperator*{\argmin}{arg\,min}

\DeclareMathOperator{\tr}{\mathbf{tr}}

\newcommand{\prl}[1]{\left(#1\right)}
\newcommand{\brl}[1]{\left[#1\right]}
\newcommand{\crl}[1]{\left\{#1\right\}}
\newcommand{\scaleMathLine}[2][1]{\resizebox{#1\linewidth}{!}{$\displaystyle{#2}$}}

\newcommand{\Revised}[2]{#2}

\input{tex/preamble.tex}

\allowdisplaybreaks

\graphicspath{{epsfiles/}}

\input{cls/sym.tex}

\usepackage{subfiles}  

\title{\LARGE \bf Control Barriers in Bayesian Learning of System Dynamics}

\author{%
  Vikas~Dhiman$^{\blue{*}}$,
  Mohammad~Javad~Khojasteh$^{\blue{*}}$,~\IEEEmembership{Member,~IEEE,}
  Massimo~Franceschetti,~\IEEEmembership{Fellow,~IEEE,}
  and~Nikolay Atanasov,~\IEEEmembership{Member,~IEEE}%
\thanks{
${\blue{^*}}$ The first two authors contributed equally.
}%
\thanks{
The material in this paper was presented in part at the 2020 Learning for Dynamics and Control Conference (L4DC) \cite{usfinal23}. We gratefully acknowledge support from ARL DCIST CRA W911NF-17-2-0181 and NSF awards CNS-1446891, ECCS-1917177, and IIS-2007141.
}%
\thanks{V. Dhiman is with Department of Electrical and Computer Engineering, University of Maine, Bangor, ME 04469. (e-mail: \texttt{vikas.dhiman@maine.edu})}
\thanks{M. J. Khojasteh is with Laboratory for Information and Decision Systems, Massachusetts Institute of Technology,  Cambridge, MA 02139.  (e-mail: \texttt{mkhojast@mit.edu})}
\thanks{M. Franceschetti, and N. Atanasov are with Department of Electrical and Computer Engineering, University of California San Diego, La Jolla, CA 92093. (e-mails: \texttt{\{mfranceschetti,natanasov\}@ucsd.edu})}
}

\begin{document}
\maketitle

\begin{abstract}
This paper focuses on learning a model of system dynamics online while satisfying safety constraints. Our objective is to avoid offline system identification or hand-specified models and allow a system to safely and autonomously estimate and adapt its own model during operation. Given streaming observations of the system state, we use Bayesian learning to obtain a distribution over the system dynamics. Specifically, we \Revised{use}{propose a new} matrix variate Gaussian process (MVGP) regression approach with an efficient covariance factorization to learn the drift and input gain terms of a nonlinear control-affine system. The MVGP distribution is then used to optimize the system behavior and ensure safety with high probability, by specifying control Lyapunov function (CLF) and control barrier function (CBF) chance constraints. We show that a safe control policy can be synthesized for systems with arbitrary relative degree and probabilistic CLF-CBF constraints by solving a second order cone program (SOCP). Finally, we extend our design to a self-triggering formulation, adaptively determining the time at which a new control input needs to be applied in order to guarantee safety.
\end{abstract}
\begin{IEEEkeywords}
Gaussian Process, learning for dynamics and control, high relative-degree system safety, control barrier function, self-triggered safe control
\end{IEEEkeywords}

\section*{Supplementary Material}
Software and videos supplementing this paper are available at: \url{https://vikasdhiman.info/Bayesian_CBF}

\input{tex/Introduction.tex}

\input{tex/RelatedWork.tex}

\input{tex/Problem.tex}

\input{tex/InferenceNew.tex}

\input{tex/SafeControl.tex}

\input{tex/ExampleRelativeDegreeTwo.tex}

\input{tex/Evaluation.tex}
\input{tex/conculsion.tex}

\appendices
\input{./tex/MGVresults.tex}

\input{./tex/mean-and-variance-of-cbf-2/matrix-variate-gp.tex}
\input{./tex/proof-MVGP-ext-lederer2019uniform}
\input{./tex/cbc-r-mean-and-variance-affine-and-quadratic.tex}

\input{./tex/mean-and-variance-of-cbf-2.tex}

\bibliographystyle{cls/IEEEtran.bst}
\bibliography{bib/main.bib}
\label{page:bibend}

\begin{IEEEbiography}
[{\includegraphics[height=1.25in,clip,keepaspectratio]{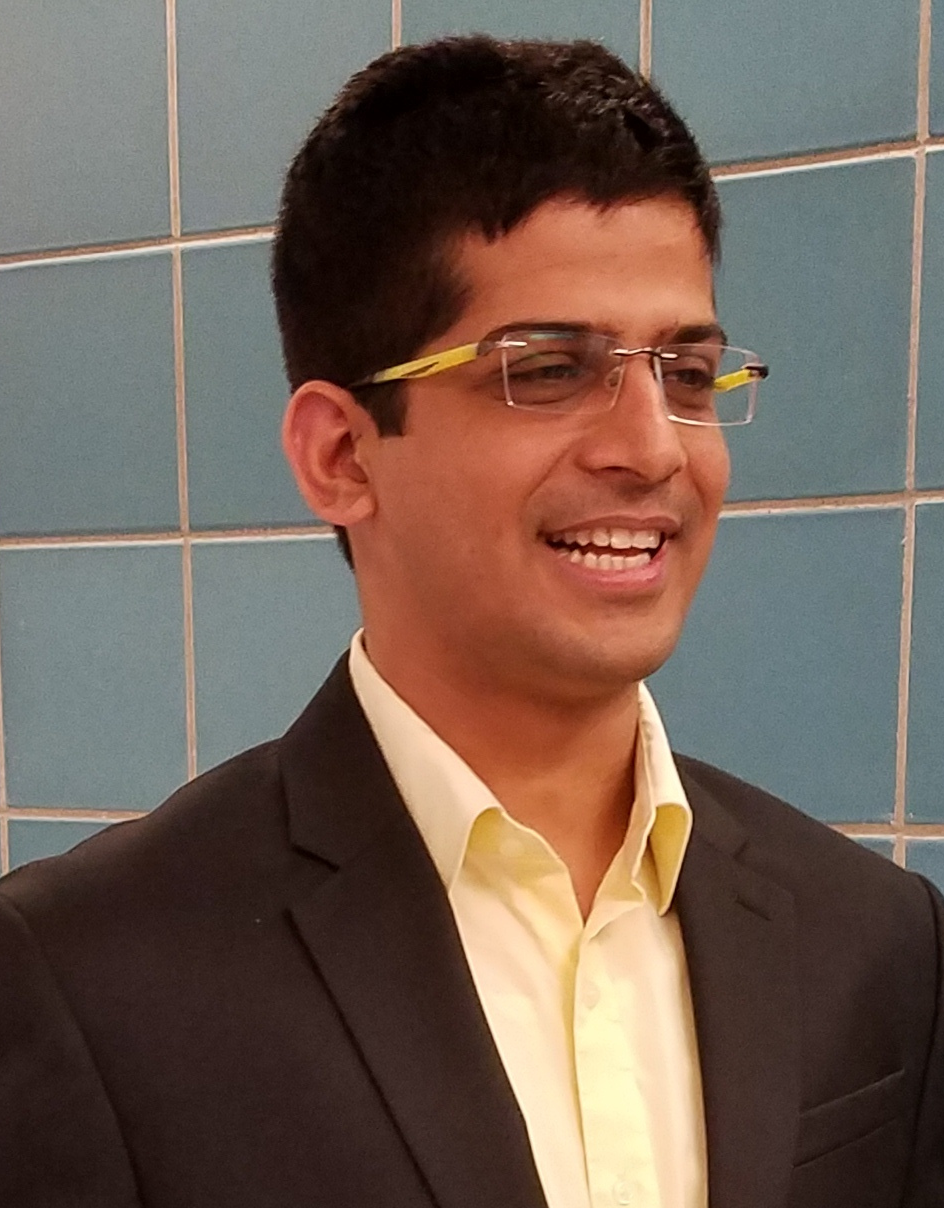}}]{Vikas
  Dhiman} is an Assistant Professor in the ECE department at the University of Maine. His works lie in the localization, mapping and control algorithms for applications in robotics. He was a Postdoctoral Researcher at the University of California, San Diego (2019-21). He graduated in Elec. Engg (2008) from Indian Institute of Technology, Roorkee, earned his MS (2014) from University at Buffalo, and received his PhD (2019) from the University of Michigan, Ann Arbor. 
\end{IEEEbiography}

\begin{IEEEbiography}
[{\includegraphics[height=1.25in,clip]{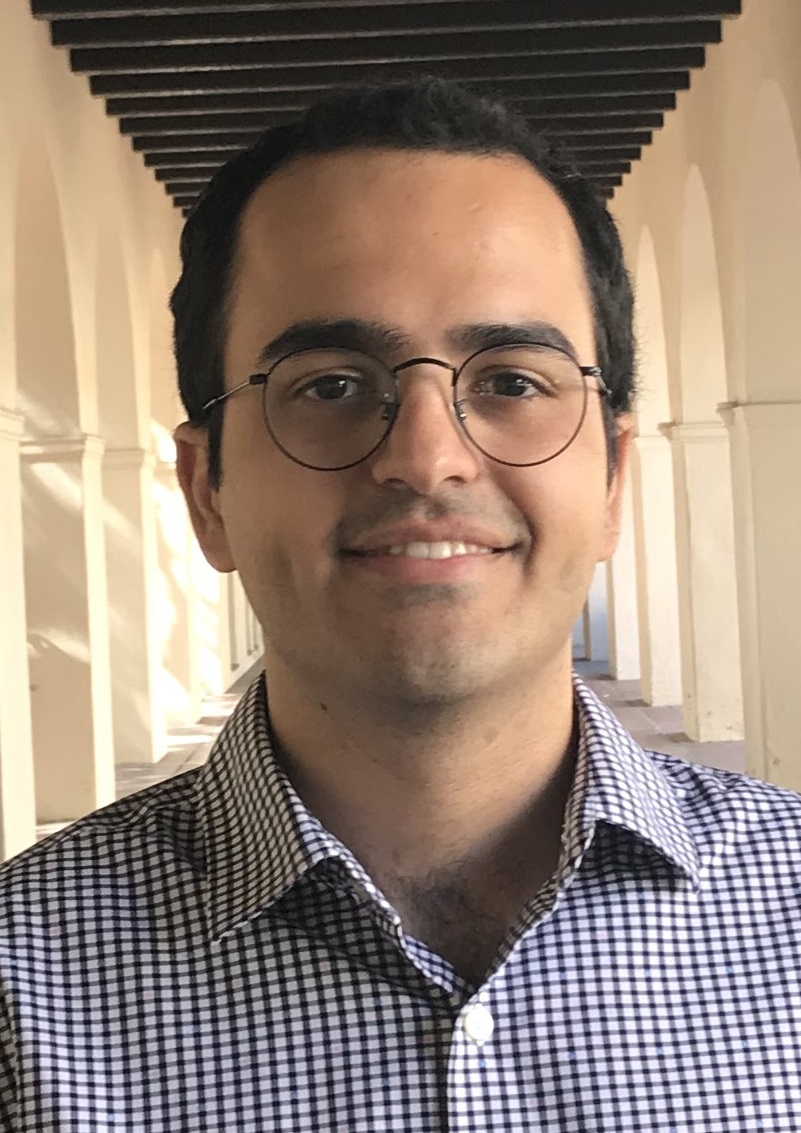}}]{Mohammad
    Javad Khojasteh}(S'14--M'21) did his undergraduate studies at  Sharif University of Technology from which he received  B.Sc.\ degrees  in both Electrical Engineering and  Mathematics, in 2015. 
  He received the M.Sc.\ and Ph.D.\ degrees in Electrical and Computer Engineering from  University of California San Diego (UCSD), La Jolla, CA, in
  2017, and 2019,
respectively. 
In 2020, he was a postdoctoral scholar with Center for Autonomous Systems and Technology (CAST) at California Institute of Technology (Caltech), and a visitor at NASA's Jet Propulsion Laboratory (JPL), where he worked with Team CoSTAR. 
Currently, he is a Postdoctoral Associate with Laboratory for Information and Decision Systems (LIDS) at Massachusetts Institute of Technology (MIT).
\end{IEEEbiography}

\begin{IEEEbiography}
  [{\includegraphics[height=1.25in,clip]{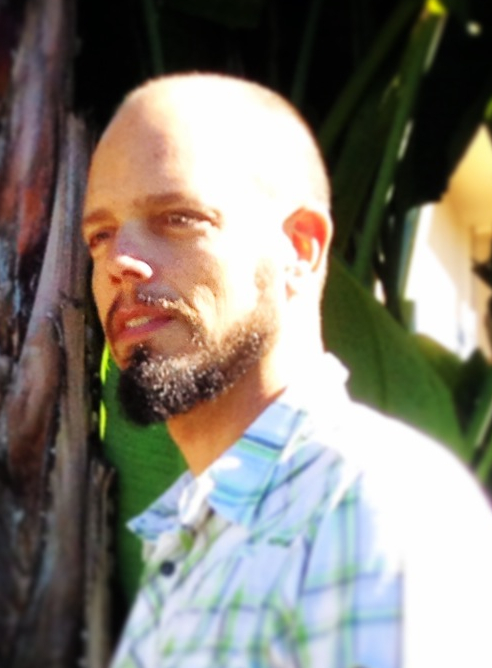}}]
  {Massimo Franceschetti} (M'98--SM'11--F'18) received the Laurea degree
  (with highest honors) in computer engineering from the University of
  Naples, Naples, Italy, in 1997, the M.S.\ and Ph.D.\ degrees in
  electrical engineering from the California Institute of Technology,
  in 1999, and 2003, respectively.  He is Professor of
  Electrical and Computer Engineering at the University of California
  at San Diego (UCSD). Before joining UCSD, he was a postdoctoral
  scholar at the University of California at Berkeley for two
  years. He has held visiting positions at the Vrije Universiteit
  Amsterdam, the \'{E}cole Polytechnique F\'{e}d\'{e}rale de Lausanne,
  and the University of Trento. His research interests are in physical
  and information-based foundations of communication and control
  systems. 
   He was awarded the C. H. Wilts
   Prize in 2003 for best doctoral thesis in electrical engineering at
   Caltech; the S.A. Schelkunoff Award in 2005 for best paper in the
   IEEE Transactions on Antennas and Propagation, a National Science
   Foundation (NSF) CAREER award in 2006, an Office of Naval Research
  (ONR) Young Investigator Award in 2007, the IEEE Communications
   Society Best Tutorial Paper Award in 2010, and the IEEE Control
   theory society Ruberti young researcher award in 2012. 
  He has been elected fellow of the IEEE in 2018 and became a Guggenheim fellow for the natural sciences: engineering in 2019.
\end{IEEEbiography}

\begin{IEEEbiography}
[{\includegraphics[width=1in,height=1.25in,clip,keepaspectratio]{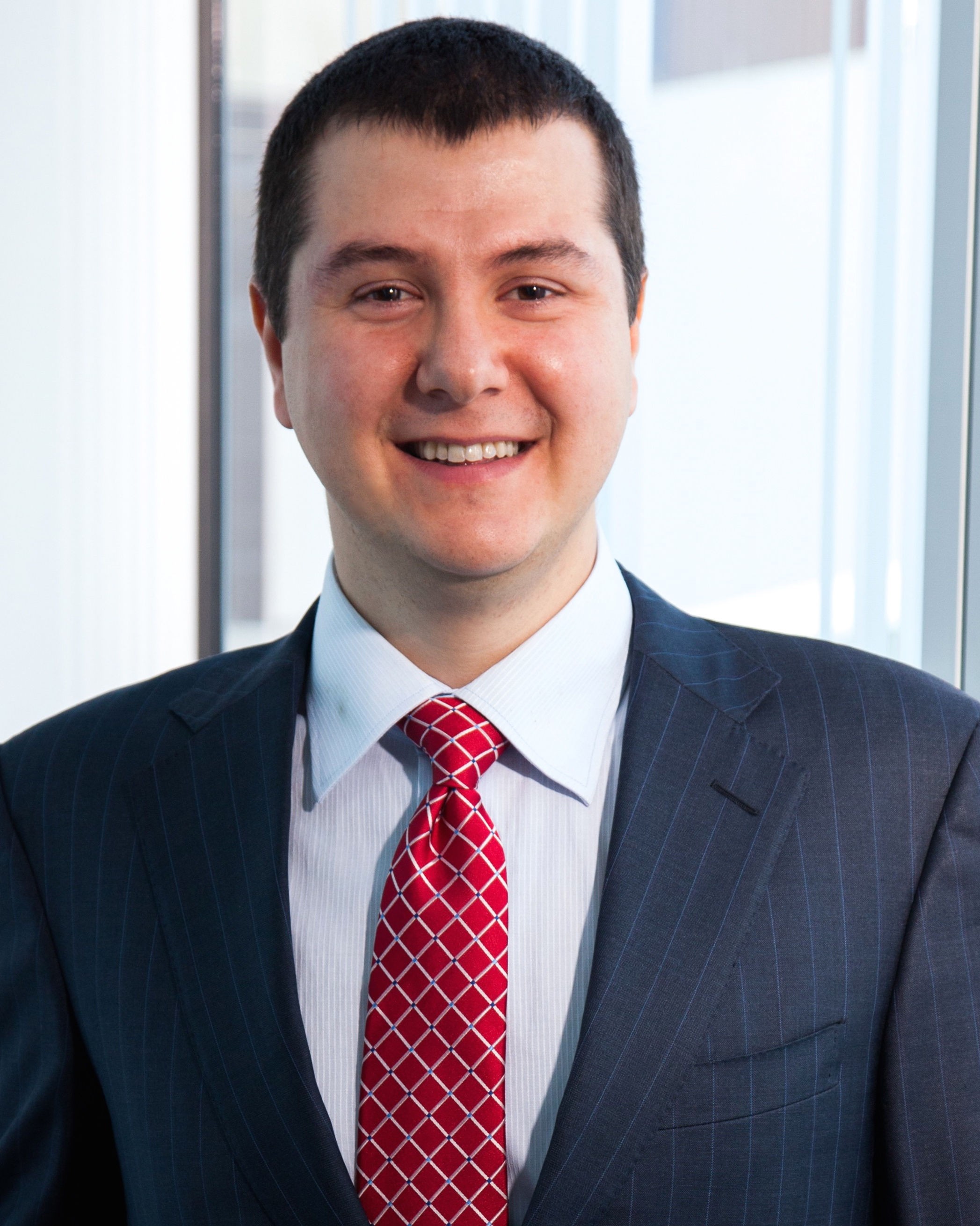}}]{Nikolay Atanasov}
(S'07-M'16) is an Assistant Professor of Electrical and Computer Engineering at the University of California San Diego. He obtained a B.S. degree in Electrical Engineering from Trinity College, Hartford, CT, in 2008 and M.S. and Ph.D. degrees in Electrical and Systems Engineering from the University of Pennsylvania, Philadelphia, PA, in 2012 and 2015, respectively. His research focuses on robotics, control theory, and machine learning, applied to active sensing using ground and aerial robots. He works on probabilistic environment models that unify geometry and semantics and on optimal control and reinforcement learning approaches for minimizing uncertainty in these models. Dr. Atanasov's work has been recognized by the Joseph and Rosaline Wolf award for the best Ph.D. dissertation in Electrical and Systems Engineering at the University of Pennsylvania in 2015 and the best conference paper award at the International Conference on Robotics and Automation in 2017.
\end{IEEEbiography}

\end{document}

%% file: tex/preamble.tex
\renewcommand{\vec}[1]{\mathbf{#1}}

\def\R{\mathbb{R}}
\def\E{\mathbb{E}}

\newcommand{\ctrl}{\vec{u}}

\newcommand{\dynAff}{F}

\newcommand{\ctrlaff}{\underline{\mathbf{\ctrl}}}

\newcommand{\knl}{\kappa}

\newcommand{\StDat}{\mathbf{X}}
\newcommand{\StDtDat}{\dot{\mathbf{X}}}
\newcommand{\CtDat}{\underline{\boldsymbol{\mathcal{U}}}_{1:k}}

\newcommand{\GP}{\mathcal{GP}}

\newcommand{\grad}{\nabla}
\newcommand{\Lie}{\mathcal{L}}

\newcommand{\CBC}{\mbox{CBC}}
\newcommand{\CBCtwo}{\CBC^{(2)}}
\newcommand{\Prob}{\mathbb{P}}

\DeclareMathOperator{\vect}{vec}
\DeclareMathOperator{\diag}{diag}
\DeclareMathOperator{\cov}{cov}

\DeclareMathOperator{\Var}{\textit{Var}}

%% file: cls/sym.tex
\newcommand{\calB}{{\cal B}}
\newcommand{\calC}{{\cal C}}

\newcommand{\calE}{{\cal E}}

\newcommand{\calG}{{\cal G}}
\newcommand{\calH}{{\cal H}}

\newcommand{\calJ}{{\cal J}}
\newcommand{\calK}{{\cal K}}

\newcommand{\calM}{{\cal M}}
\newcommand{\calN}{{\cal N}}

\newcommand{\calP}{{\cal P}}

\newcommand{\calT}{{\cal T}}
\newcommand{\calU}{{\cal U}}

\newcommand{\calX}{{\cal X}}

\newcommand{\bfb}{\mathbf{b}}
\newcommand{\bfc}{\mathbf{c}}
\newcommand{\bfd}{\mathbf{d}}
\newcommand{\bfe}{\mathbf{e}}
\newcommand{\bff}{\mathbf{f}}

\newcommand{\bfk}{\mathbf{k}}

\newcommand{\bfm}{\mathbf{m}}

\newcommand{\bfo}{\mathbf{o}}

\newcommand{\bfq}{\mathbf{q}}

\newcommand{\bfu}{\mathbf{u}}
\newcommand{\bfv}{\mathbf{v}}
\newcommand{\bfw}{\mathbf{w}}
\newcommand{\bfx}{\mathbf{x}}
\newcommand{\bfy}{\mathbf{y}}
\newcommand{\bfz}{\mathbf{z}}

\newcommand{\bfeta}{\boldsymbol{\eta}}

\newcommand{\bfmu}{\boldsymbol{\mu}}

\newcommand{\bfA}{\mathbf{A}}
\newcommand{\bfB}{\mathbf{B}}
\newcommand{\bfC}{\mathbf{C}}
\newcommand{\bfD}{\mathbf{D}}

\newcommand{\bfI}{\mathbf{I}}

\newcommand{\bfK}{\mathbf{K}}

\newcommand{\bfM}{\mathbf{M}}

\newcommand{\bfP}{\mathbf{P}}
\newcommand{\bfQ}{\mathbf{Q}}
\newcommand{\bfR}{\mathbf{R}}
\newcommand{\bfS}{\mathbf{S}}

\newcommand{\bfU}{\mathbf{U}}
\newcommand{\bfV}{\mathbf{V}}

\newcommand{\bfX}{\mathbf{X}}

\newcommand{\bfLambda}{\boldsymbol{\Lambda}}

\newcommand{\bfSigma}{\boldsymbol{\Sigma}}

\newcommand{\bbC}{\mathbb{C}}

\newcommand{\bbE}{\mathbb{E}}

\newcommand{\bbP}{\mathbb{P}}

\newcommand{\bbR}{\mathbb{R}}

\newcommand{\bbU}{\mathbb{U}}

%% file: tex/Introduction.tex
\section{Introduction}
\label{sec:intro}

Unmanned vehicles and other cyber-physical systems \cite{kumar,murray2003future} promise to transform many aspects of our lives, including transportation, agriculture, mining, and construction. 
Successful use of autonomous systems in these areas critically depends on safe adaptation in changing operational conditions. Existing systems, however, rely on brittle hand-designed dynamics models and safety rules that often fail to account for both the complexity and uncertainty of real-world operation.
Recent work \cite{bagnell2010learning,thananjeyan2020safety,deisenroth2011pilco,Dean2019,sarkar2019finite,coulson2019data,chen2018approximating,khojasteh2018learning,liu2019robust,fan2020deep,chowdhary2014bayesian,lakshmanan2020safe,levine2016end,pan2018agile,kahn2020badgr,ono2020maars} has demonstrated that learning-based system identification and control techniques may be successful at complex tasks and control objectives. 
However, two critical considerations for applying these techniques onboard autonomous systems remain less explored: \emph{learning} online, relying on streaming data, and guaranteeing \textit{safe} operation, despite the estimation errors inherent to learning algorithms. For example, consider steering a Ackermann-drive vehicle, whose dynamics are not perfectly known,  safely to a goal location. 
Not only does the Ackermann vehicle need to \emph{learn} a distribution over its dynamics from state and control observations but also account for the variance of the dynamics to \emph{guarantee safety} while executing its actions.


Several approaches have been proposed in the literature to guaratee safety of dynamical systems.
Motivated by the utility of Lyapunov functions for certifying stability properties, \cite{wieland2007constructive,ames2016control,cbf_car,xu2015robustness,prajna2007framework,cbf,barbosa2020provably,jankovic2018robust,lindemann2019control,garg2019control} proposed \textit{control barrier functions} (CBFs) as a tool to enforce safety properties in dynamical systems.
A CBF certifies whether a control policy achieves forward invariance of a \textit{safe set} $\calC$ by evaluating if the system trajectory remains away from the boundary of $\calC$. A lot of the literature on CBFs considers systems with known dynamics, low relative degree, no disturbances, and time-triggered control, in which inputs are recalculated at a fixed and sufficiently small period. Time-triggered control is limiting because low frequency may lead to safety violations in-between sampling times, while high frequency leads to inefficient use of computational resources and actuators. Yang et al. \cite{yang2019self} extend the CBF framework, for \textit{known dynamics}, to a event-triggered setup \cite{umlauft2019feedback,solowjow2020event,heemels2012introduction,hashimoto2020learning} in which the longest time until a control input needs to be recomputed to guarantee safety is provided. CBF techniques can handle nonlinear control-affine systems but many existing results apply only to relative-degree-one systems, in which the first time derivative of the CBF depends on the control input. This requirement is violated by many underactuated robot systems and motivates extensions to relative-degree-two systems, such as bipedal and car-like robots. The works \cite{hsu2015control,nguyen2016optimal,nguyen2016exponential,xiao2019control1} generalized these ideas, in the case of \textit{known dynamics}, by designing exponential control barrier function (ECBF) and high order control barrier function (HOCBF), that can handle control-affine systems with any relative degree.


Our work proposes a Bayesian learning approach for estimating the posterior distribution of the control-affine system dynamics from online data. We generalize the CBF control synthesis techniques to handle probabilistic safety constraints induced by the dynamics distribution. Our work makes the following \textbf{contributions}. First, we develop a Matrix Variate Gaussian Process (MVGP) regression to estimate the unknown system dynamics and formulate probabilistic safety and stability constraints for systems with arbitrary relative degree. Second, we show that a control policy satisfying the proposed probabilistic constraints can be obtained as the solution of a deterministic second order cone program (SOCP). The SOCP formulation depends on the mean and variance of the probabilistic safety and stability constraints, which can be obtained efficiently. Third, we extend our results to a self-triggered safe control setting, adaptively determining the duration of each control input before re-computation is necessary to guarantee safety with high probability. Finally, we derive closed-form expressions for the mean and variance of the probabilistic safety and stability constraints, up to relative-degree two. This work extends our conference paper \cite{usfinal23} by introducing a SOCP control-synthesis formulation instead of the original, possibly non-convex, quadratic program (QP) problem. This paper also extends the self-triggered control design to systems with relative degree above one, contains a complete technical treatment of the result---including the proofs of that were omitted in the conference version, and presents new evaluation results.

%% file: tex/RelatedWork.tex
\section{Related Work}

Providing safety guarantees for learning-based control techniques has received significant attention recently \cite{learning_safe_mpc,safe_bayes_opt,fisac2018general,bastani2019safe,wabersich2018safe,biyik2019efficient,jansen2018shielded,lew2020chance}. In particular, optimization-based control synthesis with CLF and CBF constraints has been considered for systems subject to additive stochastic disturbances \cite{clark2019control,santoyo2019barrier,aloysius2020safe}. Aloysius Pereira et al. \cite{aloysius2020safe} combine CBF constraints  and forward-backward stochastic differential equations,
to find a deep parameterized safe optimal control policy in the presence of additive Brownian motion in the system dynamics. \Revised{}{
Yaghoubi et al.~\cite{yaghoubi2021riskbounded} provide a stochastic CBF framework
for composite systems where a component of the system evolve according to a known deterministic dynamics and the other one follows a stochastic dynamics based on the Weiner process}. CBF conditions for systems with uncertain dynamics have been proposed in \cite{fan2019bayesian,wang2018safe,cheng2019end,marvi2020safe,taylor2019adaptive,lopez2020robust,salehi2019active,taylor2020learning}. Fan et al. \cite{fan2019bayesian} study time-triggered CBF-based control synthesis for control-affine systems with relative degree one, where the input gain is known and invertible but the drift term is learned via Bayesian techniques. The authors compare the performance of Gaussian process (GP) regression \cite{williams2006gaussian}, dropout neural networks \cite{gal2016dropout}, and ALPaCA \cite{harrison2018meta} in constructing adaptive CLF and CBF conditions using bounds on the error with respect to a system reference model. The works in \cite{wang2018safe,cheng2019end,taylor2019adaptive,lopez2020robust} study time-triggered CBF-based control of relative-degree-one systems with additive uncertainties in the drift part of the dynamics. Wang et al. \cite{wang2018safe} use GP regression to approximate the unknown part of the 3D nonlinear dynamics of a quadrotor robot. Cheng et al. \cite{cheng2019end} propose a two-layers control policy design that integrates CBF-based control with model-free reinforcement learning (RL). Safety is ensured by bounding the worst-case deviation of the dynamics estimate from the mean, using high-confidence polytopic uncertainty bounds. Marvi and Kiumarsi \cite{marvi2020safe} also consider safe model-free reinforcement learning but, instead of a two-layer policy design, the cost-to-go function is augmented with a CBF term. The works \cite{taylor2019adaptive,lopez2020robust} use adaptive CBFs to deal with parameter uncertainty. Salehi et al. \cite{salehi2019active} use Extreme Learning Machines to approximate the dynamics of a closed-loop nonlinear system (with a drift term only) subject to a barrier certificate constraint during the learning process itself.

Our work builds upon the current literature by developing a matrix variate GP regression with efficient covariance factorization to learn both the \emph{drift term} and the \emph{input gain} terms of a nonlinear control-affine system. The posterior distribution is used to ensure safety for systems with \emph{arbitrary relative degree}. Compared to previous works, our safety constraints are less conservative (probabilistic instead of worst-case) and lead to a novel SOCP formulation. We also present results for a \emph{self-triggering design} with unknown system dynamics.

Research directions left open for future investigation include the extension of the results to a frequentist setting \cite{efron}. Following \cite{safe_mbrl_nips17}, where the unknown but deterministic dynamics are assumed to belong to the Reproducing Kernel Hilbert Space (RKHS), our CBF-based self-triggered controller may be redesigned in a frequentist framework. In fact, a time-triggered setup that utlizes independent GP regression for each coordinate has been recently developed in \cite{jagtap2020control}, where a systematic approach is utilized to compute CBFs for the learned model.

%% file: tex/Problem.tex
\section{Problem Statement}
\label{sec:problem}

Consider a control-affine nonlinear system:
\begin{equation}
\label{eq:system_dyanmics}
\dot{\bfx} = f(\bfx) + g(\bfx)\bfu = \begin{bmatrix} f(\bfx) & g(\bfx)\end{bmatrix} \begin{bmatrix}1\\\bfu\end{bmatrix} \triangleq F(\bfx) \ctrlaff,
\end{equation}
where $\bfx(t) \in \calX \subset \real^n$ and $\bfu(t)\in \calU \subset \real^m$ are the system state and control input, respectively, at time $t$. Assume that the \textit{drift term} $f: \real^n \rightarrow \real^n$ and the \textit{input gain} $g: \real^n \rightarrow \real^{n \times m}$ are locally Lipschitz and the admissible control set $\calU$ is convex. We study the problem of enforcing stability and safety properties for~\eqref{eq:system_dyanmics} when $f$ and $g$ are unknown and need to be estimated online, using observations of $\bfx$, $\bfu$, $\dot{\bfx}$.

\subsection{Notation}
\label{sec:notation}
We use bold lower-case letters for vectors ($\bfx$), bold capital letters for matrices ($\bfX$), and caligraphic capital letters for sets ($\calX$). The boundary of a set $\mathcal{X}$ is denoted by $\partial \mathcal{X}$. Let $\textit{vec}(\bfX) \in \mathbb{R}^{nm}$ be the vectorization of $\bfX \in \R^{n \times m}$, obtained by stacking the columns of $\bfX$. The Kronecker product is denoted by $\otimes$. The Hessian and Jacobian of functions $h : \real^n \times \real^m \rightarrow \real$  and $f : \real^n \rightarrow \real^n$, respectively, are defined as:
\newcommand{\p}{\partial}
\begin{equation*}
\scaleMathLine{%
\begin{aligned}
  \calH_{\bfx,\bfy} h(\bfx, \bfy) &\triangleq \begin{bmatrix}
    \frac{\p^2 h}{\p x_1 \p y_1} & \dots & \frac{\p^2 h}{\p x_1 \p y_m}
    \\
    \vdots & & \vdots
    \\
    \frac{\p^2 h}{\p x_n \p y_1} & \dots & \frac{\p^2 h}{\p x_n \p y_m}
  \end{bmatrix} \;\;
    \calJ_\bfx f(\bfx) &\triangleq \begin{bmatrix}
    \frac{\p f_1}{\p x_1} & \dots & \frac{\p f_1}{\p x_n}
    \\
    \vdots & & \vdots
    \\
    \frac{\p f_n}{\p x_1} & \dots & \frac{\p f_n}{\p x_n}
    \end{bmatrix}.
\end{aligned}}
\end{equation*}
The Lie derivative of $V : \real^n \rightarrow \real$ along $f : \real^n \rightarrow \real^n$ is denoted by $\Lie_{f}V: \real^n \rightarrow \real$. The space of $r$ times continuously differentiable functions $h : \calX \rightarrow \real$ is denoted by $\bbC^r(\calX,\real)$.
%

\subsection{Stability and Safety with Known Dynamics}
\label{sec:clf-cbf-qp}

We first review key results \cite{cbf} on control Lyapunov functions for enforcing stability and control barrier functions for enforcing safety of control-affine systems with \textit{known dynamics}. System stability may be asserted as follows. 


\begin{defn}
\label{def:clf}
A function $V \in \bbC^1(\calX,\real)$ is a \emph{control Lyapunov function} (CLF) for the system in \eqref{eq:system_dyanmics} if it is positive definite, $V(\bfx) > 0$, $\forall \bfx \in \calX \setminus \{\mathbf{0}\}$, $V(\mathbf{0}) = 0$, and satisfies:
\begin{equation}
\label{eq:clc}
\inf_{\bfu \in \calU} \mbox{CLC}(\bfx,\bfu) \leq 0, \quad \forall \bfx \in \calX,
\end{equation}
where $\mbox{CLC}(\bfx, \bfu) \triangleq \Lie_{f}V(\bfx) + \Lie_{g}V(\bfx)\bfu + \gamma(V(\bfx))$ is a \emph{control Lyapunov condition} (CLC) defined for some class $K$ function $\gamma$.
\end{defn}

\begin{Proposition}[Sufficient Condition for Stability~\cite{cbf}]
\label{prop:clf}
If there exists a CLF $V(\bfx)$ for system~\eqref{eq:system_dyanmics}, then any Lipschitz continuous control policy $\pi(\bfx) \in \crl{\bfu \in \calU \mid \mbox{CLC}(\bfx,\bfu) \leq 0}$ asymptotically stabilizes the system.
\end{Proposition}

Let $\calC \triangleq \crl{\bfx \in \calX \mid h(\bfx) \geq 0}$ be a \textit{safe set} of system states, defined implicitly by a function $h \in \bbC^1(\calX,\real)$. System~\eqref{eq:system_dyanmics} is \textit{safe} with respect to $\calC$ if $\calC$ is \textit{forward invariant}, i.e., for any $\bfx(0) \in \calC$, $\bfx(t)$ remains in $\calC$ for all $t \geq 0$. System safety may be asserted as follows.
 
\begin{defn}
\label{def:cbf}
A function $h \in \bbC^1(\calX,\real)$ is a \emph{control barrier function} (CBF) for the system in~\eqref{eq:system_dyanmics} if
\begin{equation}
\label{eq:cbc}
\sup_{\bfu \in \calU} \mbox{CBC}(\bfx,\bfu) \geq 0, \quad \forall \bfx \in \calX,
\end{equation}
where $\mbox{CBC}(\bfx, \bfu) \triangleq \Lie_{f}h(\bfx) + \Lie_{g}h(\bfx)\bfu + \alpha(h(\bfx))$ is a \emph{control barrier condition} (CBC) defined for some extended class $K_\infty$ function $\alpha$.
\end{defn}

\begin{Proposition}[Sufficient Condition for Safety~\cite{cbf}]
\label{thm:Ames1}
Consider a set $\calC$ defined implicitly by $h \in \bbC^1(\calX,\real)$. If $h$ is a CBF and $\grad h(\bfx) \neq 0$ for all $\bfx$ when $h(\bfx)=0$, then any Lipschitz continuous control policy $\pi(\bfx) \in \crl{\bfu \in \calU \mid \mbox{CBC}(\bfx,\bfu) \geq 0}$ renders the system in~\eqref{eq:system_dyanmics} safe with respect to $\calC$.
\end{Proposition}


\Revised{}{Prop.~\ref{prop:clf} and Prop.~\ref{thm:Ames1} provide sufficient conditions for a control policy $\pi(\bfx)$ applied to the system in \eqref{eq:system_dyanmics} to guarantee stability and safety.} Note that the conditions are defined by affine constraints in $\bfu$. This allows the formulation of control synthesis as a quadratic program (QP) in which stability and safety properties are captured by the linear CLC and CBC constraints, respectively:
\begin{equation}
\label{eq:clf-cbf-qp}
\begin{aligned}
\pi(\bfx) \in &\argmin_{\bfu \in \calU, \delta \in \mathbb{R}} && \|\bfR(\bfx)\bfu\|^2 + \lambda \delta^2\\
&\quad\;\text{s.t.} && \mbox{CLC}(\bfx,\bfu) \le \delta,\; \mbox{CBC}(\bfx,\bfu) \ge 0,
\end{aligned}
\end{equation}
where $\bfR(\bfx) \in \mathbb{R}^{m \times m}$ is a matrix penalizing control effort and $\delta$ is a slack variable that ensures feasibility of the QP by giving preference to safety over stability, controlled by the scaling factor $\lambda > 0$. If a stabilizing control policy $\hat{\pi}(\bfx)$ is already available, it may be modified minimally online to ensure safety:
\begin{equation}
\label{eq:cbf-qp}
\begin{aligned}
\pi(\bfx) \in &\argmin_{\bfu \in \calU} &&\|\bfR(\bfx)\prl{\bfu - \hat{\pi}(\bfx)}\|^2 \\
&\quad\;\text{s.t.} && \mbox{CBC}(\bfx,\bfu) \ge 0.
\end{aligned}
\end{equation}
In practice, the QPs above cannot be solved infinitely fast. Optimization is typically performed at triggering times $t_k, t_{k+1}, \ldots$, providing control input $\bfu_k \triangleq \pi(\bfx_k)$ when the system state is $\bfx_k \triangleq \bfx(t_k)$. Ames et al. \cite[Thm.~3]{ames2016control} show that if $f$, $g$, and $\alpha \circ h$ are locally Lipschitz, then $\pi(\bfx)$ and $\mbox{CBC}(\bfx,\pi(\bfx))$ are locally Lipschitz. Thus, for sufficiently small inter-triggering times $\tau_k \triangleq t_{k+1} - t_{k}$, solving \eqref{eq:cbf-qp} at $\{t_k\}_{k \in \mathbb{N}}$ ensures safety during the inter-triggering intervals $[t_k,t_{k+1})$ as well. 

\subsection{Stability and Safety with Unknown Dynamics}
\label{Sec:problemSSUD}



This work considers stability and safety for the control-affine nonlinear system in~\eqref{eq:system_dyanmics} when the system dynamics $F(\bfx) \in \real^{n \times (1+m)}$ are \emph{unknown}. We place a prior distribution over $\vect(F(\bfx))$ using a GP \cite{williams2006gaussian} with mean function $\vect(\bfM_0(\bfx))$ and covariance function $\bfK_0(\bfx,\bfx')$.

Our objective is to compute the posterior distribution of $\vect(F(\bfx))$ using observations of the system states and controls over time and ensure stability and safety using the estimated dynamics model despite possible estimation errors. 

\begin{problem}
\label{prb:learning}
Given a prior distribution on the unknown system dynamics, $\vect(F(\bfx)) \sim \calG\calP\prl{\vect(\bfM_0(\bfx)), \bfK_0(\bfx,\bfx')}$, and a training set, $\StDat_{1:k} \triangleq [\bfx(t_1), \dots, \bfx(t_k)]$, $\bfU_{1:k} \triangleq [\bfu(t_1),\allowbreak \dots, \bfu(t_k)]$, $\StDtDat_{1:k}=[\dot{\bfx}(t_1), \dots, \dot{\bfx}(t_k)]$\footnote{If not available, the derivatives may be approximated via finite differences, e.g., $\StDtDat_{1:k} \triangleq \bigl[
  \frac{\bfx(t_2) - \bfx(t_1)}{t_2-t_1}, \dots, \frac{\bfx(t_{k+1}) -
    \bfx(t_{k})}{t_{k+1} - t_{k}} \bigr]$, provided that the inter-triggering times $\{\tau_k = t_k-t_{k-1}\}_k$ are sufficiently small.}, compute the posterior distribution $\calG\calP\prl{\vect(\bfM_k(\bfx)), \bfK_k(\bfx,\bfx')}$ of $\vect(F(\bfx))$ conditioned on $(\StDat_{1:k}, \bfU_{1:k}, \StDtDat_{1:k})$.
\end{problem}

\begin{problem}
\label{prob3444!!}
Given a safe set $\calC \triangleq \crl{\bfx \in \calX \mid h(\bfx) \geq 0}$, initial state $\bfx_k \triangleq \bfx(t_k) \in \calC$, and the distribution $\calG\calP(\textit{vec}(\bfM_k(\bfx)),\allowbreak \bfK_k(\bfx,\bfx'))$ of $\vect(F(\bfx))$ at time $t_k$, choose a control input $\bfu_k$ and triggering period $\tau_k$ such that for $\bfu(t) \triangleq \bfu_k$:
\begin{equation}
\label{eorprobmelm45}
\mathbb{P}(\mbox{CBC}(\bfx(t),\bfu_k) \ge 0) \ge p_k \quad\text{for all}\quad t \in [t_k,t_k+\tau_k),
\end{equation}
where $\bfx(t)$ follows the dynamics in~\eqref{eq:system_dyanmics}, and $p_k \in (0,1)$ is a user-specified risk tolerance.
\end{problem}

%% file: tex/InferenceNew.tex
\section{Matrix Variate Gaussian Process Regression of System Dynamics}
\label{learningsec:13452}

\newcommand{\ubcalU}{\underline{\boldsymbol{\calU}}}
\newcommand{\bcalM}{\boldsymbol{\calM}}
\newcommand{\bcalB}{\boldsymbol{\calB}}
\newcommand{\bcalC}{\boldsymbol{\calC}}

\Revised{}{This section presents our novel Matrix Variate Gaussian Process (MVGP) solution to Problem~\ref{prb:learning}.} We aim to estimate the unknown drift $f$ and input gain $g$ of system \eqref{eq:system_dyanmics} using a state-control dataset $(\StDat_{1:k}, \bfU_{1:k}, \StDtDat_{1:k})$. To simplify the derivation, we assume that $\StDat_{1:k}$ and $\bfU_{1:k}$ are observed without noise, butt each measurement $\dot{\bfx}(t_k)$ in $\StDtDat_{1:k}$ is corrupted by zero-mean Gaussian noise $\calN(\mathbf{0}, \bfS)$ that is independent across time steps. Treating $f(\bfx) + g(\bfx)\bfu$ as a single vector-valued function with input $[\bfx^\top,\bfu^\top]^\top$ may seem natural from a function approximation perspective, but this approach ignores the control-affine structure. Encoding this structure in the learning process not only increases the learning efficiency, but also ensures that the control Lyapunov and control barrier conditions remain linear in $\bfu$. Hence, we focus on learning the matrix-valued function $F(\bfx) = [f(\bfx) \; g(\bfx)]$ defined in \eqref{eq:system_dyanmics}.

The simplest approach is to use $n(1+m)$ decoupled GPs for each element of $F(\bfx)$. This approach ignores the dependencies among the components of $f(\bfx)$ and $g(\bfx)$. Furthermore, since the outputs of $f(\bfx)$ and $g(\bfx)$ are observed together via $\StDtDat_{1:k}$, training data dimensions are still mutually correlated and cannot be treated as decoupled GPs denying the efficiency advantage. At the other extreme, treating $\vect(F(\bfx))$ as a single vector-valued function and using a single GP distribution will enable high estimation accuracy but specifying an effective matrix-valued kernel function $\bfK_0(\bfx,\bfx')$ and optimizing its hyperparameters is challenging. A promising approach is offered by the Coregionalization GP (CoGP) model \cite{alvarez2012kernels}, where the kernel function is decomposed as $\bfK_0(\bfx,\bfx') \triangleq \bfSigma \knl_0(\bfx,\bfx')$ into a scalar kernel $\knl_0(\bfx,\bfx')$ and a covariance matrix parameter $\bfSigma \in \R^{n(1+m) \times (1+m)n}$. Estimating $\bfSigma$ may still require a lot of training data. Moreover, the matrix-times-scalar-kernel structure is not preserved in the posterior of the Coregionalization model, preventing its effective use for incremental learning when new data is received over time.

We propose an alternative factorization of $\bfK_0(\bfx,\bfx')$ inspired by the Matrix Variate Gaussian distribution~\cite{StructuredPBP,louizos2016structured}. The strengths of our factorization are that it models the correlation among the elements of $F(\bfx)$, preserves its structure in the posterior GP distribution, and has similar training and testing complexity as the decoupled GP approach.

\begin{defn}
The Matrix Variate Gaussian (MVG) distribution is a three-parameter distribution $\calM\calN(\bfM,\bfA,\bfB)$ describing a random matrix $\bfX \in \mathbb{R}^{n \times m}$ with probability density function:
\begin{equation*}
\scaleMathLine{p(\bfX; \bfM, \bfA, \bfB) \triangleq \frac{\exp\prl{ -\frac{1}{2}\tr\brl{\bfB^{-1}(\bfX-\bfM)^\top \bfA^{-1}(\bfX-\bfM)}}}{(2\pi)^{nm/2}\det(\bfA)^{m/2}\det(\bfB)^{n/2}}}
\end{equation*}
where $\bfM \in \mathbb{R}^{n \times m}$ is the mean. \Revised{}{Up to a scaling factor, each column of $\bfX$ has the same covariance $\bfA \in \bbR^{n \times n}$, and each row has the same covariance $\bfB \in \bbR^{m \times m}$.}
\end{defn}

Additional properties of the MVG distribution are provided in Appendix~\ref{sec:MGVresutls2}. Note that if $\bfX \sim \calM\calN(\bfM,\bfA,\bfB)$, then $\vect(\bfX) \sim \mathcal{N}(\vect(\bfM), \bfB \otimes \bfA)$. Based on this observation, we propose the following decomposition of the GP prior:
\begin{equation}
\label{eq:owggponemkr}
\vect(F(\bfx)) \sim \mathcal{GP}(\vect(\bfM_0(\bfx)), \bfB_0(\bfx,\bfx') \otimes \bfA),
\end{equation}
where $\bfB_0(\bfx,\bfx') \in \mathbb{R}^{(m+1)\times (m+1)}$ and $\bfA \in \mathbb{R}^{n \times n}$. \Revised{}{We call such a process $F(\bfx)$ with kernel structure defined in \eqref{eq:owggponemkr} a Matrix Variate Gaussian Process (MVGP).}

\begin{figure*}
    \centering
    \includegraphics[width=0.45\linewidth,trim=0mm 5mm 0mm 1mm, clip]{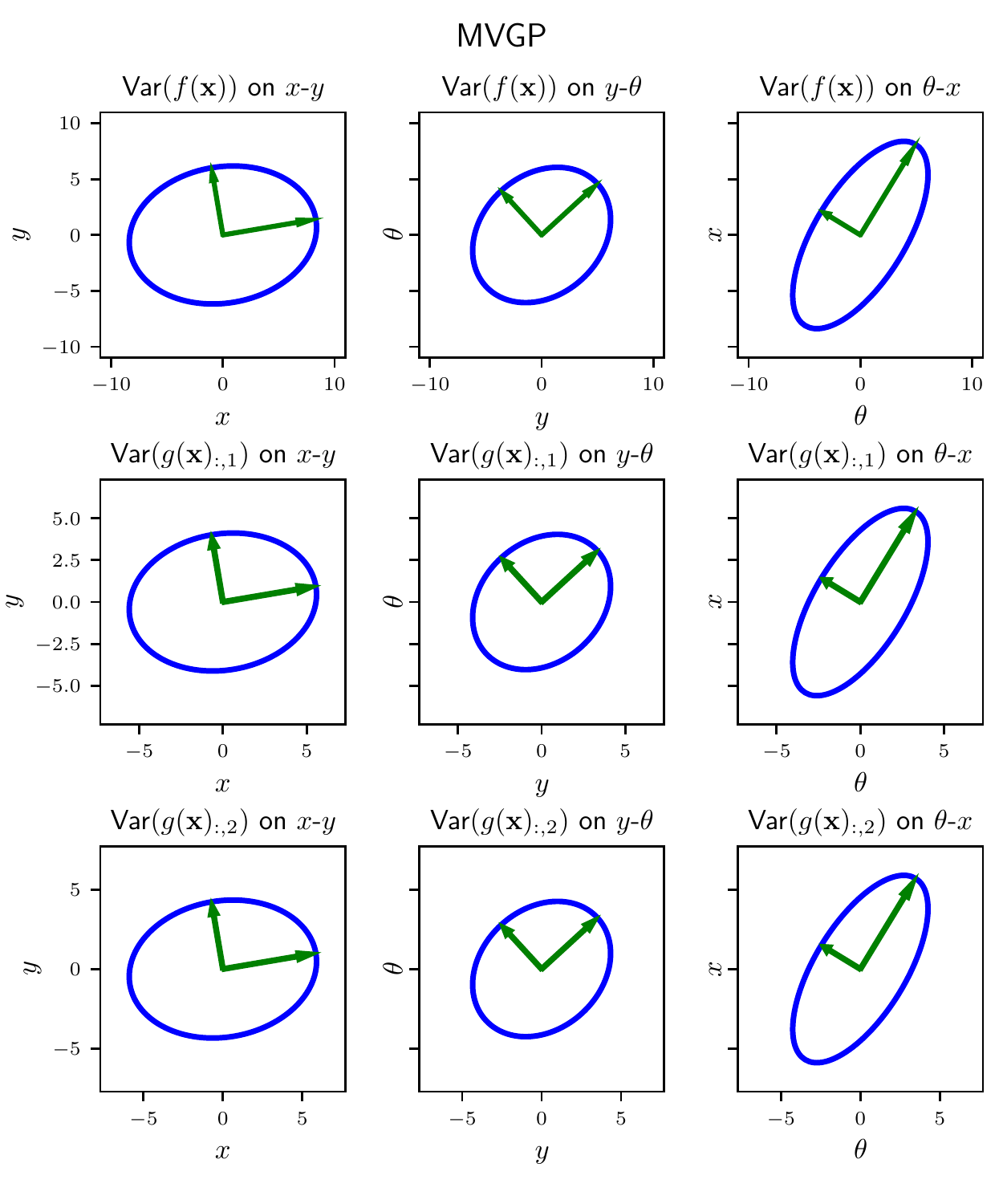}%
    \hspace{0.05\linewidth}%
    \includegraphics[width=0.45\linewidth,trim=0mm 5mm 0mm 1mm, clip]{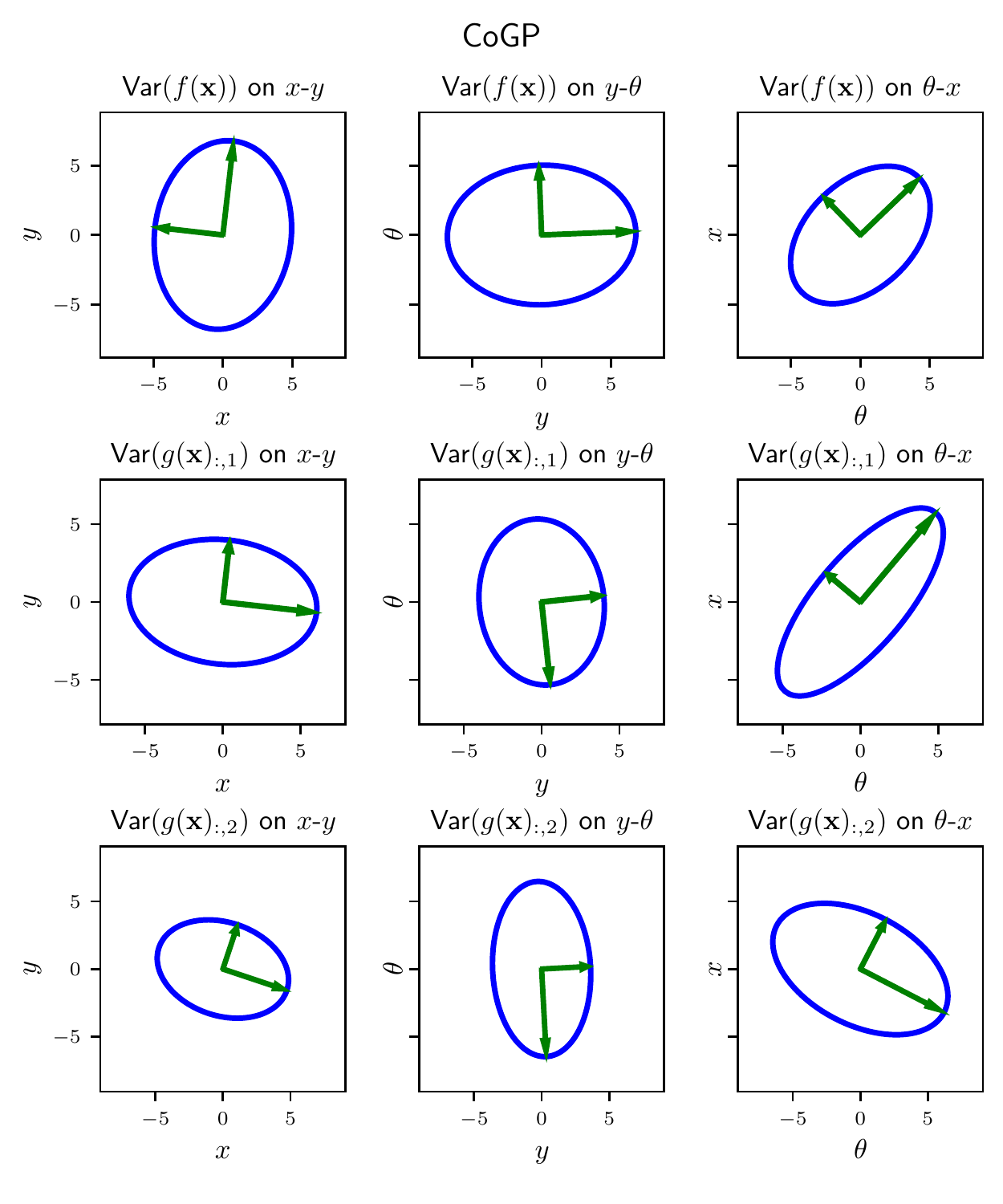}
    \caption{\Revised{}{Covariances of $f(\bfx)$ and the columns of $g(\bfx)$ learned from an Ackermann vehicle simulation described in Sec.~\ref{sec:sims-Ackermann} using an MVGP model (left) and a CoGP model (right). By construction the shapes of the covariances for each column of the MVGP model are the same (up to scale), while the column covariance shapes for the CoGP model are not constrained.}}
    \label{fig:covariances-shape}
\end{figure*}

We assume that the measurement noise covariance matrix satisfies $\bfS = \sigma^2 \bfA$ for some parameter $\sigma > 0$. This assumption on the relationship between $\bfS$ and the dynamics row covariance $\bfA$ is necessary to preserve the Kronecker product structure of the covariance in the GP posterior.


To simplify notation, let $\bcalB_{1:k}^{1:k} \in \real^{k(m+1) \times k(m+1)}$ be a matrix with $k \times k$ block elements $\brl{\bcalB_{1:k}^{1:k}}_{i,j} \triangleq \bfB_0(\bfx_i,\bfx_j)$ and define $\bcalM_{1:k} \triangleq \begin{bmatrix} \bfM_0(\bfx_1) & \cdots & \bfM_0(\bfx_k) \end{bmatrix} \in \real^{n \times k(m+1)}$, $\bcalB_{1:k}(\bfx) \triangleq \brl{\bfB_0(\bfx, \bfx_1), \dots, \bfB_0(\bfx, \bfx_k) } \in \bbR^{(m+1) \times k(m+1)}$ and $\ubcalU_{1:k}\triangleq \diag(\ctrlaff_1, \dots, \ctrlaff_k) \in \R^{k(m+1) \times k}$. Consider an arbitrary test point \Revised{$\bfx$}{$\bfx_*$}. The train and test data are jointly Gaussian:
\begin{multline}
\label{eq:joint-training-test-distribution}
  \begin{bmatrix}
    \vect(\StDtDat_{1:k})
    \\
    \vect(F(\bfx_*))
  \end{bmatrix}
  \sim
  \mathcal{N}\left(
  \begin{bmatrix}
    \vect(\bcalM_{1:k}\ubcalU_{1:k})
    \\
    \vect(\bfM_0(\bfx_*))
  \end{bmatrix},
  \right.\\
  \left.
  \begin{bmatrix}
    \ubcalU_{1:k}^\top \bcalB_{1:k}^{1:k} \ubcalU_{1:k} + \sigma^2\bfI_k
    & \ubcalU_{1:k}^\top\bcalB_{1:k}^\top(\bfx_*)
    \\
    \bcalB_{1:k}(\bfx_*)\ubcalU_{1:k} &  \bfB_0(\bfx_*, \bfx_*)
  \end{bmatrix} \otimes \bfA  \right),
\end{multline}
where
\begin{align}
  &\ubcalU_{1:k}^\top \bcalB_{1:k}^{1:k} \ubcalU_{1:k}
  =\begin{bmatrix}
    \underline{\bfu}_1^\top \bfB_0(\bfx_1,\bfx_1) \underline{\bfu}_1& \cdots & \underline{\bfu}_1^\top \bfB_0(\bfx_1,\bfx_k) \underline{\bfu}_k
    \\ \vdots&\ddots&\vdots\\
    \underline{\bfu}_k^\top \bfB_0(\bfx_k,\bfx_1) \underline{\bfu}_1& \cdots & \underline{\bfu}_k^\top \bfB_0(\bfx_k,\bfx_k) \underline{\bfu}_k
  \end{bmatrix}
  \notag\\
    &\bcalB_{1:k}(\bfx_*)\ubcalU_{1:k}
    =
    \begin{bmatrix}
      \bfB_0(\bfx_1,\bfx_*)\underline{\bfu}_1
      &
      \cdots
      &
      \bfB_0(\bfx_k,\bfx_*)\underline{\bfu}_k
    \end{bmatrix}.
\end{align}
Applying a Schur complement to \eqref{eq:joint-training-test-distribution}, we find the distribution of $\vect(F(\bfx_*))$ conditioned \Revised{on}{on the training data $(\StDat_{1:k}, \bfU_{1:k}, \StDtDat_{1:k})$. The posterior mean and covariance functions are provided in the next proposition.}

\begin{Proposition}
\label{thmfrmvg24}
The posterior distribution of $\vect(F(\bfx))$ conditioned on the training data
$(\StDat_{1:k}, \bfU_{1:k}, \StDtDat_{1:k})$ is a Gaussian process
$\mathcal{GP}(\vect(\bfM_k(\bfx)), \bfB_k(\bfx,\bfx')\otimes\bfA)$ with
parameters:
\begin{gather}
\label{eq:mvg-posterior}
\begin{aligned}
  &\bfM_k(\bfx) \triangleq \bfM_0(\bfx) +
  \prl{ \dot{\bfX}_{1:k} - \bcalM_{1:k}\ubcalU_{1:k}} \prl{\ubcalU_{1:k}\bcalB_{1:k}(\bfx)}^\dagger
  \\
  &\bfB_k(\bfx,\bfx') \triangleq \bfB_0(\bfx,\bfx') 
  -
  \bcalB_{1:k}(\bfx)\ubcalU_{1:k} \prl{\ubcalU_{1:k}\bcalB_{1:k}(\bfx')}^\dagger
  \\
  &\prl{\ubcalU_{1:k}\bcalB_{1:k}(\bfx)}^\dagger
  \triangleq
  \prl{\ubcalU_{1:k}^\top\bcalB_{1:k}^{1:k}\ubcalU_{1:k} + \sigma^2 \bfI_k}^{-1}\ubcalU_{1:k}^\top\bcalB_{1:k}^\top(\bfx).
\end{aligned}
\raisetag{7ex}
\end{gather}
\end{Proposition}

\begin{proof}
See Appendix~\ref{proof:matrix-variate-gp}.
\end{proof}


\Revised{Thm}{Prop}.~\ref{thmfrmvg24} shows that, for a given control input $\bfu$, the MVGP posterior of $F(\bfx)\underline{\bfu}$ is (applying Lemma~\ref{lemma:product-mvg}):
\begin{equation}
F(\bfx)\underline{\bfu}
\sim \mathcal{GP}(\bfM_k(\bfx)\underline{\bfu}, \underline{\bfu}^\top\bfB_k(\bfx,\bfx')\underline{\bfu}\otimes\bfA).
\label{eq:gp_posterior}
\end{equation}

A key property of the MVGP model in \Revised{Thm}{Prop}.~\ref{thmfrmvg24} is that the posterior preserves the kernel decomposition. In the special case of, $\bfB_0(\bfx,\bfx') \triangleq \bfB \knl_0(\bfx,\bfx')$, in addition to the kernel parameters, our model requires learning of $\bfB \in \mathbb{R}^{(1+m)\times (1+m)}$ and, if the measurement noise statistics are unknown, $\bfA \in \mathbb{R}^{n \times n}$ and $\sigma \in \bbR$. Thus, the MVGP model has fewer parameters than a CoGP model, which requires $(1+m)^2n^2$ parameters for the full covariance decomposition $\bfK_0(\bfx,\bfx') \triangleq \bfSigma \knl_0(\bfx,\bfx')$. For a data set with $k$ examples, the training computational complexity of our MVGP approach is $O(k^3)$, while the same for the CoGP model is $O(k^3n^3)$. This is because our model inverts only a $(k \times k)$ matrix while the CoGP model inverts a $(kn \times kn)$ matrix (see Appendix~\ref{proof:coregionalization-gp}). \Revised{}{Fig.~\ref{fig:covariances-shape} shows a comparison of the covariances obtained from an MVGP model (with $\bfB_0(\bfx,\bfx') = \bfB \knl_0(\bfx,\bfx')$) and from a CoGP model. The shapes of the MVGP covriances for different columns of $F(\bfx)$ are the same by construction, while those of CoGP model may be arbitrary. Thus, the CoGP model is more expressive than our MVGP model. However, MVGP model offers considerable computational advantages especially for high-dimensional systems without significant loss in accuracy. The accuracy and computational complexity of the MVGP and CoGP models are compared in Sec.~\ref{sec:eg-pendulum}.}

%% file: tex/SafeControl.tex
\section{Self-triggered Control with Probabilistic Safety Constraints}
\label{sec:unknowndynamics}
\newcommand{\CBCr}{\mbox{CBC}_k^{(r)}}

The MVGP model developed in Sec.~\ref{learningsec:13452} addresses Problem~\ref{prb:learning} and provides a probabilistic estimate of the unknown system dynamics $F(\bfx)\underline{\bfu}$ in \eqref{eq:gp_posterior}. In this section, we consider Problem \ref{prob3444!!}. We extend the optimization-based control synthesis approach in \eqref{eq:clf-cbf-qp} to handle probabilistic stability and safety constraints induced by the posterior distribution of $F(\bfx)\underline{\bfu}$. We focus on handling probabilistic safety constraints of the form $\mathbb{P}(\mbox{CBC}(\bfx,\bfu) \ge 0) \ge p$. Since the control Lyapunov and control barrier conditions have the same form, involving the Lie derivative of a known function $V(\bfx)$ or $h(\bfx)$ with respect to the system dynamics, our approach can be used for stability constraints as well, as demonstrated in \Revised{}{Sec. \ref{sec:sims-Ackermann}}. Furthermore, we develop self-triggering conditions to adaptively decide the times $t_k$, $t_{k+1}$, $\ldots$, at which the system inputs $\bfu$ should be recomputed in order to guarantee safety with high probability along the continuous-time system trajectory, instead of instantaneously in time.

\Revised{}{Our approach requires the stochastic system dynamics $F(\bfx)\underline{\bfu}$ to be locally Lipschitz continuous with high probability.
This is necessary to ensure that for any MVGP sample $\hat{F}(\bfx)$, the ordinary differential equation $\dot{\bfx} = \hat{F}(\bfx)\underline{\bfu}$ in \eqref{eq:system_dyanmics} has a unique solution \cite[Thm.~3.20]{sastry_book}. Without such an assumption, the system would be able to instantaneously move from a safe state with a large safety margin to an unsafe state, making it impossible to guarantee safety for any time duration after a control triggering instant.}
\Revised{\begin{ansatz}
\label{ass:LipschitzGP}
Let $\Prob_k$ be the probability measure induced by the distribution of $F(\bfx)\underline{\bfu}$ at time $t_k$.
Assume that for any $L_k > 0$, $\bfu$, and $\tau_k$, there exists a constant $b_k > 0$ such that \Revised{}{the stochastic system dynamics $F(\bfx(t))\ctrlaff$ is locally Lipschitz continuous with high probability $q_k$:}
\begin{align}
\label{eq:smoth23}
\Revised{\Prob_k}{\Prob}\biggl(\sup_{s \in [0, \tau_k)} &\|F(\bfx(t_k+s))\ctrlaff_k-F(\bfx(t_k))\ctrlaff_k\|\\
& \le L_k \|\bfx(t_k+s)-\bfx(t_k)\| \biggr) \ge q_k \triangleq \Revised{1-e^{-b_kL_k}}{1-\delta_{L_k}}. 
\notag
\end{align}
\end{ansatz}}
{The next Lemma shows constructively how to compute a Lipschitz constant $L_\bff$ which holds with a user-specified probability $1-\delta_L$ for a vector-valued GP model  with  kernel whose fourth-order partial derivatives are Lipschtz continuous. The proof follows, with minor modifications, the argument in \cite[Thm 3.2]{lederer2019uniform}.

\begin{lemma}
\label{lem:LipschitzGP}
Consider the vector-valued function $\bff(\bfx) \triangleq F(\bfx)\ubfu$ in \eqref{eq:gp_posterior} with distribution $\GP(\bfmu(\bfx), \bfA \knl_\bfB(\bfx, \bfx'))$, where $\bfmu(\bfx) = \bfM_k(\bfx)\ubfu$ and $\knl_\bfB(\bfx, \bfx') = \ubfu^\top \bfB_k(\bfx, \bfx') \ubfu$. Suppose that $\knl_\bfB(\bfx, \bfx')$ has continuous partial derivatives of fourth order. For $j \in \{1, \dots, n\}$, let $\knl^{\partial j}_\bfB(\bfx, \bfx') \triangleq \frac{\partial^2}{\partial \bfx_j \partial \bfx'_j} \knl_\bfB(\bfx, \bfx')$ with Lipschitz constant $L_\knl^{\partial j}$ over a compact set $\calX$ with maximal extension $r = \max_{\bfx, \bfx' \in \calX} \|\bfx - \bfx'\|$. Then, each element of the Jacobian of $\bff(\bfx)$ is bounded with probability of at least $1-\frac{\delta_L}{n^2}$:
\begin{equation}
    \left|\sup_{\bfx \in \calX} \frac{\partial \bff_i(\bfx)}{\partial \bfx_j}\right| \le L_{i,\partial j}, \;\; \forall i, j \in \{1, \dots, n \},
\end{equation}
where
\[
\scaleMathLine{L_{i, \partial j} \triangleq \sqrt{2 \log\left(\frac{2n^2}{\delta_L}\right)} \knl^{i,\partial j}_{\calX} + 12 \sqrt{6 n} \max \left\{ \knl^{i,\partial j}_{\calX}, \sqrt{r \bfA_{ii} L_\knl^{\partial j}} \right\}}
\]
and $\knl^{i,\partial j}_{\calX} \triangleq \max_{\bfx \in \calX} \sqrt{\bfA_{ii}\knl^{\partial j}_\bfB(\bfx, \bfx')}$. Furthermore, a sample function $\bff(\bfx)$ of the given GP is almost surely continuous on $\calX$ and with probability of at least $1-\delta_L$, $L_\bff \triangleq \sqrt{\frac{1}{n^2}\sum_{i=1}^n \sum_{j=1}^n L_{i,\partial j}^2}$ is a Lipschitz constant of $\bff(\bfx)$ on $\calX$.
\end{lemma}

\begin{proof}
See Appendix~\ref{sec:proof-MVGP-lederer2019uniform}.
\end{proof}

The assumptions of Lemma~\ref{lem:LipschitzGP} are mild and reasonable. The assumption about kernel continuity can be satisfied by an appropriate choice of a kernel function. For example, commonly used radial basis functions are infinitely differentiable and continuous. The assumption that the space $\calX$ is bounded is reasonable for physical systems. For example, the region that an Ackermann vehicle can cover can be bounded based on the maximum velocity and acceleration of the vehicle.
}

\subsection{Probabilistic Safety Constraints}


Consider probabilistic stability and safety constraints in the CLF-CBF QP in \eqref{eq:clf-cbf-qp}, induced by the MVGP distribution of $F(\bfx)\underline{\bfu}$ at time $t_k$ in \eqref{eq:gp_posterior}:
\begin{equation}
\label{prog:CBF-CLF:unknown}
\begin{aligned}
\pi(\bfx_k) &\in \argmin_{\bfu_k \in \calU, \delta \in \bbR}
\|\bfR(\bfx_k)\bfu_k\|^2 + \lambda \delta^2\\
\text{s.t.}~~&\mathbb{P}(\mbox{CLC}(\bfx_k,\bfu_k) \le \delta \mid \bfx_k,\bfu_k) \ge \tilde{p}_k\\
&\mathbb{P}(\mbox{CBC}(\bfx_k,\bfu_k) \ge \zeta \mid \bfx_k,\bfu_k) \ge \tilde{p}_k,
\end{aligned}
\end{equation}
where \Revised{}{$\tilde{p}_k \triangleq p_k/(1-\delta_L)$ and $\delta_L$ is specified in Lemma~\ref{lem:LipschitzGP}.

\begin{remark}
Given a desired safety probability $p_k$ for the interval $[t_k,t_k+\tau_k)$, defined in \eqref{eorprobmelm45}, we choose the Lipschitz continuity probability $1-\delta_L$ in Lemma~\ref{lem:LipschitzGP} to ensure instantaneous safety at time $t_k$ with higher probability $\tilde{p}_k \in [0.5, 1)$. The choice of $\delta_L$ affects the Lipschitz constant $L_\bff$ and in turn the duration $\tau_k$, as discussed in Sec.~\ref{slfcc11}. By ensuring safety at time instant $t_k$ with probability at least $\tilde{p}_k$, we guarantee safety with probability at least $p_k=\tilde{p}_k(1-\delta_L)$ during the interval $[t_k,t_k+\tau_k)$ (see Proposition~\ref{thm:validity-inter-triggeringt-times}). Also see Remark \ref{rem:feasil2} for further discussion on feasibility. 
\oprocend
\end{remark}}

To ensure that the safety constraint does not only hold instantaneously at $t_k$ but over an interval $[t_k, t_k + \tau_k)$, we enforce a tighter constraint via $\zeta > 0$ and determine the time $\tau_k$ for which it remains valid. The choice of $\zeta$ and its effect on $\tau_k$ are discussed next. \Revised{}{First, we obtain the mean and variance of the $\mbox{CBC}$ constraint, which is an affine transformation of the system dynamics.}

\begin{lemma}
\label{lem:NRcalculation}
Consider the dynamics in~\eqref{eq:system_dyanmics} with posterior distribution in~\eqref{eq:gp_posterior}. Given $\bfx_k$ and $\bfu_k$, $\mbox{CBC}_k \triangleq \mbox{CBC}(\bfx_k,\bfu_k)$ is a Gaussian random variable with the following parameters:
\begin{align}
\label{eq:parametofpi5543}
\E[\mbox{CBC}_k] &= \nabla_\bfx h(\bfx_k)^\top \bfM_k(\bfx_k)\underline{\bfu}_k + \alpha(h(\bfx_k)),\\
\Var[\mbox{CBC}_k] &=  (\underline{\bfu}_k^\top\bfB_k(\bfx_k,\bfx_k)\underline{\bfu}_k)
(\nabla_\bfx h(\bfx_k)^{\top}\bfA\nabla_\bfx h(\bfx_k)). \notag
\end{align}
\end{lemma}

\begin{proof}
We start by rewriting the definition of CBC as:
\begin{align}
  \mbox{CBC}_k&= \nabla_\bfx h(\bfx_k)^\top F(\bfx_k)\ctrlaff_k + \alpha(h(\bfx_k)).
\end{align}
Then, the mean and variance of $\mbox{CBC}_k$ are:
\begin{equation}
\begin{aligned}
\E[\mbox{CBC}_k] &= \nabla_\bfx h(\bfx_k)^\top \E[F(\bfx_k)\underline{\bfu}_k] + \alpha(h(\bfx_k)),\\
\Var[\mbox{CBC}_k] &=  
\nabla_\bfx h(\bfx_k)^{\top}\Var[F(\bfx_k)\ctrlaff_k]\nabla_\bfx h(\bfx_k),
\end{aligned}
\end{equation}
where the mean and variance of $F(\bfx_k)\ctrlaff_k$ are from \eqref{eq:gp_posterior}.
\end{proof}

Using Lemma~\ref{lem:NRcalculation}, we can rewrite the safety constraint as
\begin{equation}
\label{userspecife422}
\mathbb{P}(\mbox{CBC}_k \ge \zeta | \bfx_k,\bfu_k) = 1 - \Phi\prl{\frac{\zeta-\E[\mbox{CBC}_k]}{\sqrt{\Var[\mbox{CBC}_k]}}} \geq \tilde{p}_k, 
\end{equation}
where $\Phi(\cdot)$ is the cumulative distribution function of the standard Gaussian.
Note that if the control input is chosen so that $\frac{\zeta-\E[\mbox{\tiny{CBC}}_k]}{\sqrt{\Var[\mbox{\tiny{CBC}}_k]}} \to -\infty$, as the posterior variance of $\mbox{CBC}_k$ tends to zero, the probability $\mathbb{P}(\mbox{CBC}_k \ge \zeta | \bfx_k,\bfu_k)$ tends to one.
Namely, as the uncertainty about the system dynamics tends to zero, our results reduce to the setting of Sec.~\ref{sec:clf-cbf-qp}, and safety can be ensured with probability one. Using \eqref{userspecife422} and noting that $\Phi^{-1}(1-\tilde{p}_k) = - \sqrt{2}\mbox{erf}^{-1}(2\tilde{p}_k-1)$\Revised{}{, where $\mbox{erf}(\cdot)$ is the Gauss error function,} the optimization problem in \eqref{prog:CBF-CLF:unknown} can be restated as:
\begin{equation}
\label{eq:cbf-clf-exp-var}
\begin{aligned}
    \pi(\bfx_k) &\in \argmin_{\bfu_k \in \calU, \delta \in \bbR}\; \|\bfR(\bfx_k)\bfu_k\|^2 + \lambda \delta^2\\
    \text{s.t. } 
    & \delta - \E[ \mbox{CLC}_k] \ge  c(\tilde{p}_k)\sqrt{\Var[\mbox{CLC}_k]}\\
    &\zeta-\E[\mbox{CBC}_k] \le -c(\tilde{p}_k)\sqrt{\Var[\mbox{CBC}_k]},
\end{aligned}
\end{equation}
where $c(\tilde{p}_k) \triangleq \sqrt{2}\mbox{erf}^{-1}(2\tilde{p}_k-1)$. 
\Revised{}{Using the reformulation in \eqref{eq:cbf-clf-exp-var}, the next proposition shows that the control problem with probabilistic stability and safety constraints in \eqref{prog:CBF-CLF:unknown} is a convex optimization problem, which can be solved efficiently.}

\begin{Proposition}
\label{prop:socp-rel-deg-1}
The chance-constrained optimization problem in \eqref{prog:CBF-CLF:unknown} for control-affine system dynamics in \eqref{eq:system_dyanmics} with posterior distribution \eqref{eq:gp_posterior} is a second-order cone program (SOCP):
\begin{subequations}
\begin{align}
\pi(\bfx_k) &\in \argmin_{y \in \bbR, \delta \in \bbR, \bfu \in \calU}  \quad y
\\
~~\text{ s.t. }~~
&y - \left\| \bfR(\bfx_k) \bfu + \sqrt{\lambda} \delta \right\|_2 \ge 0
\\
  &
    \bfq_k(\bfx_k)^\top \ubfu - \delta
    +
c(\tilde{p}_k) \left\| \bfP_k(\bfx_k)\ubfu \right\|_2
    \le 0
    \label{eq:stochastic-clc-rd-1}\\
&
  \bfe_k(\bfx_k)^\top \ubfu - \zeta
- c(\tilde{p}_k) \left\| \bfV_k(\bfx_k)\ubfu \right\|_2
  \ge 0,
  \label{eq:stochastic-cbc-rd-1}
\end{align}
\label{lower:nu1}
\end{subequations}
where  
\begin{align*}
 \bfq_k(\bfx_k) &\triangleq \bfM_k(\bfx_k)^\top \grad_\bfx V(\bfx_k)  + \begin{bmatrix} \gamma(V(\bfx_k)) \\ \mathbf{0}_m \end{bmatrix},
 \\
 \bfP_k(\bfx_k) &\triangleq \sqrt{\grad_\bfx V(\bfx_k)^\top \bfA \grad_\bfx V(\bfx_k)} \bfB_k^{\frac{1}{2}}(\bfx_k, \bfx_k),
 \\
 \bfe_k(\bfx_k) &\triangleq \bfM_k(\bfx_k)^\top \nabla_\bfx h(\bfx_k) + \begin{bmatrix}\alpha(h(\bfx_k)) \\ \mathbf{0}_m\end{bmatrix},
 \\
 \bfV_k(\bfx_k) &\triangleq \sqrt{\grad_\bfx h(\bfx_k)^\top \bfA \grad_\bfx h(\bfx_k)}\bfB_k^{\frac{1}{2}}(\bfx_k, \bfx_k),
\end{align*}
and $\bfB_k^\frac{1}{2}(\bfx_k, \bfx_k)$ is the Cholesky factorization of $\bfB_k(\bfx_k, \bfx_k)$. 
\end{Proposition}

\begin{proof}
Substituting the mean and variance of $\mbox{CLC}_k$ and $\mbox{CBC}_k$ (Lemma~\ref{lem:NRcalculation}) in \eqref{eq:cbf-clf-exp-var} and introducing an auxiliary variable $y$ to express the quadratic objective as a SOCP constraint leads to the result.
\end{proof}

\Revised{}{Prop.~\ref{prop:socp-rel-deg-1} generalizes the CLF-CBF QP formulation in \eqref{eq:clf-cbf-qp}. It shows that when the stability and safety constraints are probabilistic, induced by the MVGP distribution of the system dynamics, the control synthesis problem becomes a SOCP and, hence, is still convex.} We refer to the left-hand side of the constraints in \eqref{eq:stochastic-cbc-rd-1} and \eqref{eq:stochastic-clc-rd-1} as Stochastic CBC (SCBC) and Stochastic CLC (SCLC), respectively. Note that when $\tilde{p}_k = 0.5$, $c(\tilde{p}_k) = 0$ and the SOCP in \eqref{lower:nu1} reduces to the deterministic-case QP in \eqref{eq:clf-cbf-qp}. Otherwise, the SOCP can be solved to arbitrary precision $\epsilon$ within $O(-\log(\epsilon))$ iterations \cite{lobo1998applications} with off-the-shelf solvers \cite{optimization2014inc,vandenberghe2010cvxopt}.

\begin{remark}
\label{rem:feasil2}
\Revised{}{
If the SOCP program in \eqref{lower:nu1} is infeasible, there does not exist a control input that guarantees safety for the current dynamics GP distribution and the user-specified probabilistic safety requirement. This can be used as a failure mode to ask a human to take over control or provide additional training data that reduces the dynamics GP uncertainty. Since our SOCP constraint is a conservative approximation of the true probabilistic safety constraint in \eqref{prog:CBF-CLF:unknown}, it is possible that control that meets the desired safety bounds exists even if the SOCP is infeasible. Please refer to Caste\~nada et al~\cite{castaneda2021pointwise} for discussion of feasibility conditions of an SOCP controller.
}
\oprocend
\end{remark}

\subsection{Self-triggering Design}
\label{slfcc11}

The safety constraint in \eqref{lower:nu1} ensures safety with high probability at the triggering times $\{t_k\}_{k \in \mathbb{N}}$. Here, we extend our analysis to the inter-triggering invervals $[t_k,t_k + \tau_k)$. \Revised{}{We restate \cite[Prop.~1]{yang2019self} in our setting, showing that when the system dynamics are Lipschitz continuous (Lemma~\ref{lem:LipschitzGP}), then the change in the state is bounded.}



\begin{lemma}
\label{self-trige3Lip}
\Revised{}{Consider the GP model of the system dynamics in \eqref{eq:gp_posterior} with zero-order hold control for $[t_k,t_k+\tau_k)$. If the assumptions of Lemma~\ref{lem:LipschitzGP} are satisfied, i.e., with high probability the system dynamics are Lipschitz continuous with Lipschitz constant $L_{\bff_k}$, then for all $s \in [0, \tau_k)$ the total change in the system state is bounded with probability at least $1-\delta_L$}:
\begin{align}\label{defnirrrt5654}
   &\|\bfx(t_k+s)-\bfx_k\| \le  \overline{r}_k(s)\triangleq\frac{1}{L_{\bff_k}}\|\dot{\bfx}_k\|\left(e^{sL_{\bff_k}}-1\right).
\end{align}
\end{lemma}

Recall from Sec.~\ref{sec:clf-cbf-qp} that $h$ is a continuously differentiable function. Thus, using Lemma~\ref{self-trige3Lip}, for any inter-triggering time $\tau_k$, there exist a constant $L_{h_k} > 0$ such that:
\begin{align}
\label{eq:grad-h-bounded}
    \sup_{s \in [0, \tau_k)}\|\nabla_\bfx h(\bfx(t_k+s))\| \le L_{h_k}.
\end{align}
This is used in the next \Revised{theorem}{proposition} which concerns Problem~\ref{prob3444!!}. \Revised{}{The next proposition guarantees safety during the inter-triggering intervals as long as the controller ensures safety at the triggering time using the SOCP in~\eqref{lower:nu1}.}

\begin{Proposition}
\label{thm:validity-inter-triggeringt-times}
\Revised{}{
Consider the GP model of the system dynamics in \eqref{eq:gp_posterior} with safe set $\mathcal{C}$. Assume the system dynamics have Lipschitz constant $L_{\bff_k}$ with probability at least $1-\delta_L$ (Lemma~\ref{lem:LipschitzGP}), the system is safe with margin $\zeta$ at triggering time $t_k$ with probability at least $\tilde{p}_k = p_k / (1-\delta_L)$ (Prop.~\ref{prop:socp-rel-deg-1}), and $\alpha$ is Lipschitz continuous with Lipschitz constant $L_\alpha$. Then, the system is safe according to the constraint in \eqref{eorprobmelm45} with probability at least $p_k$ for time duration $\tau_k \le \frac{1}{L_{\bff_k}}\ln\left(1+\frac{L_{\bff_k}\zeta}{(L_{\bff_k}+L_\alpha)  L_{h_k}\|\dot{\bfx}_k\|}\right)$, where $L_{h_k}$ is defined in~\eqref{eq:grad-h-bounded}.}
\end{Proposition}


\begin{proof}
Since the program~\eqref{prog:CBF-CLF:unknown} has a solution at the triggering time $t_k$, we know $\mathbb{P}(\mbox{CBC}_k \ge \zeta | \bfx_k,\bfu_k) \ge \tilde{p}_k$.
Thus, conditioned
on $\mbox{CBC}_k \ge \zeta$ and the Lipschitz continuity of the system dynamics, if we prove $\mbox{CBC}(s+t_k) \ge 0$, for all $s \in [0, \tau_k)$, the result follows. \Revised{}{Using the Cauchy-Schwarz inequality, Lipschitz continuity of the system dynamics and Lipschitz continuity of $\alpha$ and $h$ (from \eqref{eq:grad-h-bounded}), for all $s \in [0, \tau_k)$, we have with probability at least $\tilde{p}_k (1-\delta_L)$}:
\begin{align}\label{nigrer45654e!}
    &|\mbox{CBC}(\bfx(s+t_k),\bfu_k)-\mbox{CBC}_k|\\
   &\quad\le \left(\sup_{s \in [0, \tau_k)}\|\nabla_\bfx h(\bfx(s+t_k))\|L_{\bff_k} + L_\alpha L_{h_k} \right)\overline{r}_k(s).\notag
\end{align}
Thus, using \eqref{defnirrrt5654} and \eqref{eq:grad-h-bounded}, we notice that the right-hand side of \eqref{nigrer45654e!} is upper-bounded by  $\frac{L_{h_k}L_{\bff_k}+L_{\alpha}L_{h_k}}{L_{\bff_k}}\|\dot{\bfx}_k\|\left(e^{L_{\bff_k} s}-1\right),$
which, in turn, is less than or equal to $\zeta$ for $s=\tau_k$.
\end{proof}


\Revised{}{Prop.~\ref{prop:socp-rel-deg-1} ensures instantaneous safety with probability $\tilde{p}_k \geq p_k$, while Prop.~\ref{thm:validity-inter-triggeringt-times} extends the safety guarantee with probability $p_k$ to the interval $[t_k, t_k+\tau_k)$. System safety as formalized in Problem~\ref{prob3444!!} is guaranteed because $\bbP(\mbox{CBC}(\bfx) \ge 0) \ge p_k$ implies $\bbP(h(\bfx) \ge 0) \ge p_k$, for all $t \in [t_k, t_k + \tau_k)$, due to Prop.~\ref{thm:Ames1}. Prop.~\ref{thm:validity-inter-triggeringt-times} characterizes the longest time $\tau_k$ until it is necessary to recompute the control input to guarantee safety.}


\Revised{}{Fig.~\ref{fig:triggering-time} shows an evaluation of the Lipschitz constant $L_{\bff_k}$ (Lemma~\ref{lem:LipschitzGP}) and triggering time $\tau_k$ (Prop.~\ref{thm:validity-inter-triggeringt-times}) for a system with Ackermann dynamics. We compute $L_{\bff_k}$ and $\tau_k$ both analytically and numerically. The shapes of the curves for the numerical and analytical Lipschitz constant are the same, although the analytical bound in Lemma~\ref{lem:LipschitzGP} is a few orders of magnitude more conservative. More details about the simulation are presented in Sec.~\ref{sec:sims-Ackermann}.}

\begin{figure}
  \includegraphics[width=0.9\linewidth,trim=0mm 0mm 0mm 5mm, clip]{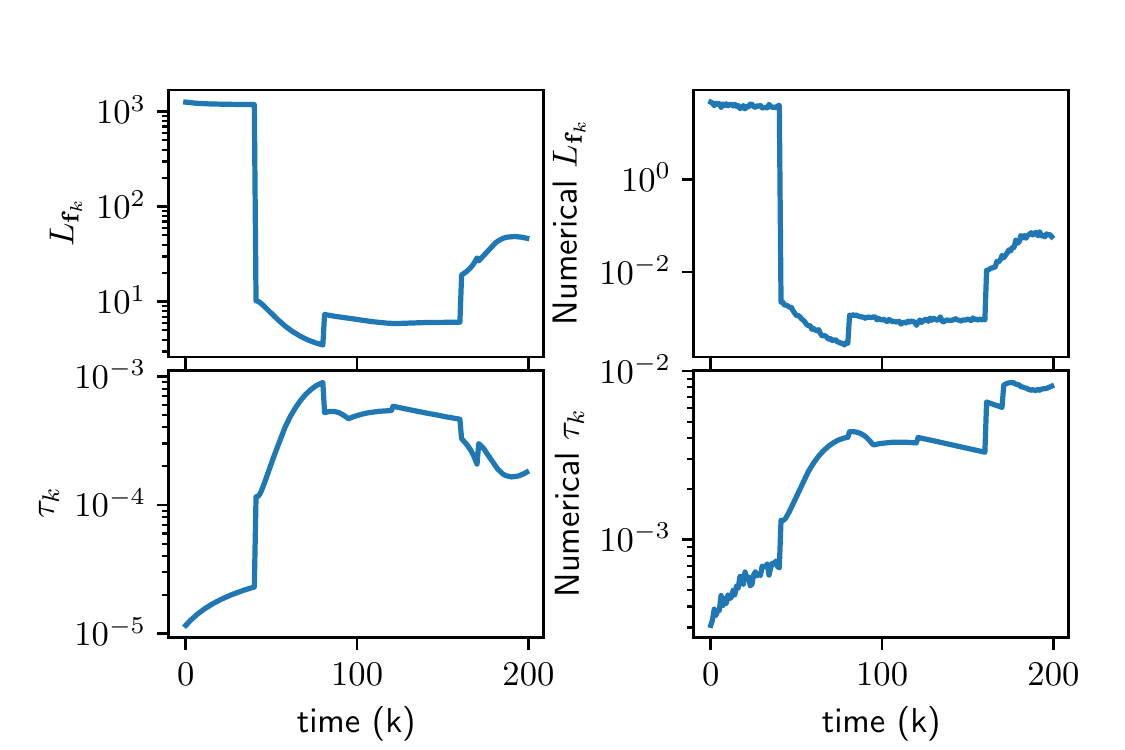}
  \caption{\Revised{}{Lipschitz constant $L_{\bff_k}$ and triggering time $\tau_k$ for an Ackermann vehicle simulation. The analytical values (left) obtained from Lemma~\ref{lem:LipschitzGP} and Prop.~\ref{thm:validity-inter-triggeringt-times} are compared with numerical estimates (right). The shapes are similar for both terms over a 200 time-step simulation but the magnitudes differs. See Sec.~\ref{sec:self-triggering-design-sim} for details.}}
  \label{fig:triggering-time}
\end{figure}

\section{Extension to Higher Relative-degree Systems}
\label{sec:exp-cbf-rel-deg-2}
\newcommand{\tdLie}{\hat{\Lie}}
\newcommand{\bLie}{\bar{\Lie}}
\newcommand{\bfxdot}{\dot{\bfx}}
\newcommand{\ff}{\mathfrak{\zeta}}


Next, we extend the probabilistic safety constraint formulation for systems with arbitrary relative degree\footnote{The motivation for assuming known relative degree but unknown dynamics comes from robotics applications. Commonly, the class of the system is known
but the parameters (e.g., mass, moment of inertia) and high-order interactions (e.g., jerk, snap, drag) of the dynamics are unknown. Estimating the relative degree is left for future work (see \cite{akella2020formal,robey2020learning,srinivasan2020synthesis}).}, using an exponential control barrier function (ECBF)~\cite{nguyen2016exponential,cbf}. Sec.~\ref{highrel304kno} reviews ECBF results for systems with \emph{known dynamics}. Sec.~\ref{MMnei23343} investigates probabilistic ECBF constraints, while Sec.~\ref{fei999411} provides a self-triggering formulation.


\subsection{Known Dynamics}
\label{highrel304kno}
Let $r \ge 1$ be the relative degree of \eqref{eq:system_dyanmics} with respect to $h(\bfx)$, i.e., $\Lie_g\Lie_f^{(r-1)}h(\bfx) \ne 0$ and $\Lie_g\Lie_f^{(k-1)}h(\bfx) = 0$, $\forall k \in \Revised{\{1, \dots, r-2\}}{\{1, \dots, r-1\}}$. This condition can be expressed by defining
\begin{equation*}
    \bfeta(\bfx) \triangleq \begin{bmatrix} h(\bfx) \\ \Lie_fh(\bfx)
      \\ \vdots\\ \Lie_f^{(r-1)}h(\bfx) \end{bmatrix},
    ~
    \mathcal{F} \triangleq \begin{bmatrix} 0  & 1 & \dots & 0 \\
      \vdots & \vdots &  & \vdots \\
      0 & 0 & \dots  & 1 \\
      0 & 0 & \dots & 0
    \end{bmatrix},
    ~
    \mathcal{G} \triangleq \begin{bmatrix} 0 \\ \vdots \\ 0 \\1 \end{bmatrix}
\end{equation*}
and $h(\bfx)$ is the output of the linear time-invariant system:
\begin{align}
\label{eq:traverse-system}
& \dot{\bfeta}(\bfx) = \mathcal{F}\bfeta(\bfx) + \mathcal{G}\bfu,
& h(\bfx) = \bfc^\top\bfeta(\bfx),
\end{align}
where $\bfc \triangleq [1, 0, \dots, 0]^\top \in \R^r$.


\begin{defn}[{\cite{nguyen2016exponential,cbf}}]
\label{def:cbf-high}
A function $h \in \bbC^r(\calX,\real)$ is an exponential control barrier function (ECBF) for the system in~\eqref{eq:system_dyanmics} if there exists a vector $\bfk_\alpha \in \R^r$ such that the $r$-th order \emph{control barrier condition} $\mbox{CBC}^{(r)}(\bfx, \bfu) \triangleq \Lie_f^{(r)}h(\bfx) + \Lie_g\Lie_f^{(r-1)}h(\bfx)\bfu+\bfk_\alpha^\top \bfeta(\bfx)$ satisfies $\sup_{\bfu \in \calU} \mbox{CBC}^{(r)}(\bfx, \bfu) \geq 0$ for all $ \bfx \in \mathcal{X}$ and $h(\bfx(t)) \ge \bfc^\top\bfeta(\bfx_0)e^{(\mathcal{F}-\mathcal{G}\bfk_\alpha)t}  \ge 0$, whenever $h(\bfx(0)) \ge 0$.
\end{defn}

If $\bfk_\alpha$ satisfies the properties in \cite[Thm.~2]{nguyen2016exponential}, then any controller $\bfu = \pi(\bfx)$ that ensures $\mbox{CBC}^{(r)}(\bfx, \bfu) \ge 0$ renders the dynamics \eqref{eq:system_dyanmics} safe with respect to $\mathcal{C} =\crl{\bfx \in \calX \mid h(\bfx) \geq 0}$. For relative-degree-one systems, $\bfk_\alpha^\top \bfeta(\bfx)$ reduces to $\alpha h(\bfx)$ with $\alpha>0$. Thus, $\mbox{CBC}^{(1)}$, the safety condition based on ECBF for $r=1$, is equivalent to the CBC in Def.~\ref{def:cbf} when the extended class $K_\infty$ function $\alpha$ is linear.

\subsection{Probabilistic ECBF Constraints}
\label{MMnei23343}

As in \eqref{prog:CBF-CLF:unknown}, we are interested in solving
\begin{equation}
\label{progg3456755!!}
\begin{aligned}
\pi(\bfx_k) &\in \argmin_{\bfu_k \in \calU}  \|\bfR(\bfx_k)\bfu_k\|_2 
\\
\text{s.t.}~~ &\mathbb{P}(\mbox{CBC}_k^{(r)} \ge \zeta | \bfx_k,\bfu_k) \ge \tilde{p}_k,
\end{aligned}
\end{equation}
where $\mbox{CBC}_k^{(r)} \triangleq \mbox{CBC}^{(r)}(\bfx_k,\bfu_k)$ and we have dropped the CLC term for clarity of presentation. While we explicitly characterised the distribution of $\mbox{CBC}_k$ for relative degree one in Lemma~\ref{lem:NRcalculation}, the distribution of $\mbox{CBC}_k^{(r)}$ cannot be determined explicitly for arbitrary $r$. Instead, we use a concentration bound to rewrite the chance constraint in terms of the moments of $\mbox{CBC}_k^{(r)}$. \Revised{}{Before we can derive a SOCP controller for higher relative-degree systems, we need to show that the mean and variance of $\mbox{CBC}_k^{(r)}$ are affine and quadratic in $\bfu$, respectively.}

\begin{lemma}
\label{prob:VarUquadrac24562}
Let $r \ge 1$ be the relative degree of \eqref{eq:system_dyanmics} with respect to $h(\bfx) \in \bbC^r(\calX,\bbR)$ and $F(\bfx)$ be a random function with finite mean and variance. Then, the expectation and variance of $\mbox{CBC}^{(r)}(\bfx,\bfu)$ in Def.~\ref{def:cbf-high} are linear and quadratic in $\bfu$, respectively, and satisfy:
\begin{equation}
\begin{aligned}
\E[\mbox{CBC}^{(r)}(\bfx, \bfu)] &= \bfe^{(r)}(\bfx)^\top\ctrlaff,\\
\Var[\mbox{CBC}^{(r)}(\bfx, \bfu)] &= \ctrlaff^\top \bfV^{(r)}(\bfx) \bfV^{(r)}(\bfx)^\top \ctrlaff,
\end{aligned}
\end{equation}
where
\begin{align}
  \bfe^{(r)}(\bfx) 
  &\triangleq
    \E\left[\dynAff(\bfx)^\top \grad_\bfx \Lie_f^{(r-1)}h(\bfx) 
    + 
    \begin{bmatrix}\bfk_\alpha^\top \bfeta(\bfx)
    \\
    \mathbf{0}_{m}\end{bmatrix}
    \right]
    \\
    \bfV^{(r)}(\bfx) 
    &\triangleq
    \Var\left[
    \dynAff(\bfx)^\top \grad_\bfx \Lie_f^{(r-1)}h(\bfx) 
    +
    \begin{bmatrix}\bfk_\alpha^\top \bfeta(\bfx)
    \\
    \mathbf{0}_{m}\end{bmatrix}
    \right]^{\frac{1}{2}}
\end{align}
\end{lemma}
\begin{proof}
See Appendix~\ref{proof:cbc-r-mean-affine-var-quadratic}.
\end{proof}


\Revised{}{The next proposition generalizes Prop.~\ref{prop:socp-rel-deg-1} to arbitrary relative-degree systems by showing that control synthesis with probabilistic safety constraints can still be cast as a SOCP.}

\begin{Proposition}
\label{prop:456723o433302}
Consider the control-affine system in \eqref{eq:system_dyanmics} \Revised{}{with stochastic dynamics} and relative degree $r \ge 1$ with respect to an ECBF $h(\bfx) \in \bbC^r(\calX,\bbR)$. If the system inputs are determined in continuous time $t_k$ from the following SOCP,
\begin{align}
\label{progg!!5544i}
&\pi(\bfx_k) \in \argmin_{y \in \bbR, \bfu_k \in \calU } \quad y
\notag\\
&\text{s.t.}\quad y - \left\|\bfR(\bfx_k) \bfu_k  \right\|_2 \ge 0
\\\notag
& \qquad
\bfe^{(r)}_k(\bfx_k)^\top\ctrlaff_k -\zeta -
c^{(r)}(\tilde{p}_k)
\left\| 
\bfV^{(r)}_k(\bfx_k) \ctrlaff_k
\right\|_2
\ge 0,
\end{align}
where $c^{(r)}(\tilde{p}_k) \triangleq \sqrt{\frac{\tilde{p}_k}{1-\tilde{p}_k}}$ and $\tilde{p}_k \in (0,1)$ is fixed for all $t_k$, then the system trajectory remains in the safe set $\calC = \crl{\bfx \in \calX \mid h(\bfx) \geq 0}$ with probability $\tilde{p}_k$.
\end{Proposition}


\begin{proof}
See Appendix~\ref{proof:cbc-r-quad-prog}.
\end{proof}

In the proposition above, we bound $\mathbb{P}(\mbox{CBC}_k^{(r)} \ge \zeta | \bfx_k,\bfu_k)$ in \eqref{progg3456755!!} using Cantelli's inequality, because the exact distribution of $\CBCr$ is unknown. The safety constraint in \eqref{progg!!5544i} can be interpreted as: ``The system must satisfy $\mbox{CBC}_k^{(r)}$ in expectation by a margin $c(\tilde{p}_k)$ times the standard deviation of $\mbox{CBC}_k^{(r)}$''. Solving the SOCP requires knowledge of the mean and variance of $\CBCr$. Sec.~\ref{sec:rdt} shows how to determine the mean and variance of $\mbox{CBC}^{(2)}_k$, which is of practical relevance for many physical systems. For $r>2$, Monte Carlo sampling can be used to estimate these quantities. 



\subsection{Self-triggering Design}
\label{fei999411}

In this section, we extend the time-triggered formulation for probabilistic safety of high relative-degree systems to a self-triggering setup. Our extension to self-triggering control in the relative-degree-one case in Proposition~\ref{thm:validity-inter-triggeringt-times}, relied on Lipshitz continuity of the CBC components. To simplify the analysis for arbitrary relative degree $r$, we temporarily introduce a modified CBF $h_b(\bfx) \triangleq h(\bfx)-\zeta_b$ with $\zeta_b > 0$. We solve the SOCP in \eqref{progg!!5544i} with $\zeta=0$ but replace $h(\bfx)$ with $h_b(\bfx)$. Enforcing the safety constraint in \eqref{progg!!5544i} for $h_b(\bfx_k)$, ensures that with probability $\tilde{p}_k$, $h(\bfx_k) \ge \zeta_b$ at the sampling times. We find an upper bound on the sampling time $\tau_k$ that ensures $h(\bfx(t))$ remains non-negative during the inter-triggering intervals $[t_k, t_k+\tau_k)$. \Revised{}{As shown in the next proposition, this is sufficient to guarantee safety during the inter-triggering interval as long as the system using the SOCP controller~\eqref{progg!!5544i} is safe at the triggering time $t_k$.}


\begin{Proposition}
\label{ir5o45t345}
\Revised{}{Consider the system in \eqref{eq:system_dyanmics} with stochastic dynamics and safe set $\mathcal{C} = \crl{\bfx \in \calX \mid h(\bfx) \geq 0}$. Suppose the system has relative degree $r \geq 1$ with respect to $h(\bfx)$ and the SOCP in \eqref{progg!!5544i} has a solution at triggering time $t_k$ with $h_b(\bfx) \triangleq h(\bfx)-\zeta_b$ as the CBF and $\zeta=0$. Suppose that the system dynamics are Lipschitz continuous with probability at least $1-\delta_L$ (Lemma~\ref{lem:LipschitzGP}) and $h$ is Lipschitz continuous with Lipschitz constant $L_{h_k}$.
%
%
Then, 
$h(\bfx(t)) \ge 0$, for $t \in [t_k, t_k+\tau_k)$, holds 
with probability $p_k=\tilde{p}_k(1-\delta_L)$, where $\tau_k \le \frac{1}{L_{\bff_k}}\ln\left(1+\frac{L_{\bff_k}\zeta_b}{L_{h_k}\|\dot{\bfx}_k\|}\right)$.}
\end{Proposition}

\begin{proof}
Using Lemma~\ref{lem:LipschitzGP} and~\ref{self-trige3Lip}  we have,
\begin{align}
    \sup_{s \in [0, \tau_k)} &| h(\bfx(t_k+s))-h(\bfx_k)| \le L_{h_k} \|\bfx(t_k+s)-\bfx_k\|
    \notag\\
    &\le \frac{L_{h_k}}{L_{\bff_k}} \|\dot{\bfx}_k\|(e^{L_{\bff_k}\tau_k}-1) \le \zeta_b,
\end{align}
where the last inequality follows from the upper bound on $\tau_k$. Given $h(\bfx(t)) \ge \zeta_b$ at $t_k$, we deduce $h(\bfx(t)) \ge 0$ for all time in $[t_k,t_k+\tau_k)$, and the result follows.
\end{proof}

Assuming that the Lipschitz continuity of system dynamics occurs with probability at least $1-\delta_L$ (Lemma~\ref{lem:LipschitzGP}), and given the value of $\|\dot{\bfx}_k\|$ at triggering time $t_k$ and the parameter $\zeta_b > 0$ and the Lipschitz constant $L_{h_k}$, Proposition~\ref{ir5o45t345} characterizes the longest time $\tau_k$ until a control input needs to be recomputed to guarantee safety with probability at least $p_k=\tilde{p}_k(1-\delta)$ for a system with arbitrary relative degree.

\begin{remark}
\label{remar288y}
Recent
works~\cite{xiao2019control1,sarkar2020high,yaghoubi2020training}studya more general high-order CBF (HOCBF) formulation 
than ECBF. Applying
our probabilistic safety constraints
with a HOCBF is left for future work.
\oprocend
\end{remark}


%% file: tex/ExampleRelativeDegreeTwo.tex
\section{Special Case: Relative Degree Two}
\label{sec:rdt}

\newcommand{\mCBCtwo}{\E[\CBCtwo(\bfx, \bfu)]}
\newcommand{\VCBCtwo}{\Var[\CBCtwo(\bfx, \bfu)]} 
\newcommand{\postf}{\bff_k}
\newcommand{\LieOne}{\Lie_{f}}
\newcommand{\LieTwo}{\Lie_{F(\bfx)\ubfu}^{(2)}}
\newcommand{\LieTwoF}{\Lie_{f}^{(2)}}
\newcommand{\LieTwoG}{\Lie_{g} \LieOne}
\newcommand{\bfzero}{\mathbf{0}}
\newcommand{\mutdf}{{\bfm}_{f}}
\newcommand{\ktdf}{\bfB_k[1,1]}
\newcommand{\mutdgu}{{\bfM}_{g}}
\newcommand{\ktdgu}{{\bfB}_{g }}
\newcommand{\mutdbfu}{{\mu}_{\bff}}
\newcommand{\ktdbfu}{{\knl}_{\bff }}
\newcommand{\ktdfgu}{{\knl}_{f,g}}
\newcommand{\ktdfbfu}{\bfb_k^\top}
\newcommand{\bfone}{\mathbf{1}}
\newcommand{\tdmu}{\hat{\mu}}
\newcommand{\tdbfk}{\hat{\bfk}}
\newcommand{\hatf}{f}
\newcommand{\hDynAff}{\dynAff_k}


Systems with relative degree two appear frequently. We develop an efficient approach for computing the mean and variance of the relative-degree-two safety condition:
\begin{equation*}
\CBCtwo(\bfx, \bfu) = [\grad_\bfx \LieOne h(\bfx)]^\top \dynAff(\bfx) \ctrlaff +  [h(\bfx), \Lie_{f}h(\bfx)]^\top\bfk_\alpha,
\end{equation*}%
which is needed to specify the safety constraints in our SOCP formulation in \eqref{progg!!5544i}.
Note that $\CBCtwo(\bfx, \bfu)$ is a stochastic process whose distribution is induced by the GP $\vect(F(\bfx))$. Fig.~\ref{fig:mean-and-variance-of-cbf-r-diagram} shows a computation graph for $\CBCtwo(\bfx, \bfu)$ with random or deterministic functions as nodes and computations among them as edges.
The graph contains only two kinds of edges: (1) dot product of a deterministic function with a vector-variate GP and (2) gradient of a scalar GP. We show how the mean, variance, and covariance propagate through these two computations in Lemma~\ref{lemma:dot-prod-gps} (dot product) and Lemma~\ref{lemma:differentiating-gp} (scalar GP gradient).

\begin{figure}
\resizebox{\linewidth}{!}{\input{tex/mean-and-variance-of-cbf-r-diagram.tex}}
\caption{Computation graph for the second-order Lie derivative of a control barrier function $h(\bfx)$ along the GP-distributed dynamics $\dynAff(\bfx) \ctrlaff$ of a relative degree two system. Only two kinds of operations are required: (1) dot product with a vector-variate GP ({\color{red}$[\cdot]$}) and (2) gradient of a scalar GP ({\color{blue}$[\grad_\bfx]$}).}
\label{fig:mean-and-variance-of-cbf-r-diagram}
\end{figure}
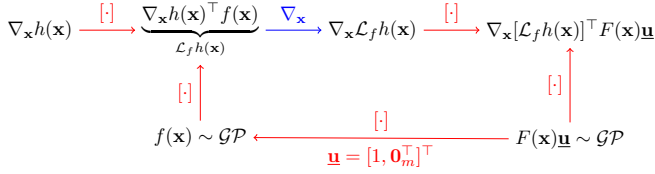

\begin{lemma}
\label{lemma:dot-prod-gps}
Consider three Gaussian random vectors $\bfx$, $\bfy$, $\bfz$ with means $\bar{\bfx}$, $\bar{\bfy}$, $\bar{\bfz}$ and variances $\Var[\bfx]$, $\Var[\bfy]$, $\Var[\bfz]$, respectively. Let their pairwise covariances be $\cov(\bfx, \bfy)$, $\cov(\bfy, \bfz)$ and $\cov(\bfz, \bfx)$. The mean, variance, and covariances of $\bfx^\top \bfy$ are given by:
\begin{align}
&\E[\bfx^\top \bfy] = \bar{\bfx}^\top \bar{\bfy} + \frac{1}{2}\tr(\cov(\bfx, \bfy) + \cov(\bfy, \bfx))\\
&\Var[\bfx^\top \bfy] = \frac{1}{2}\tr(\cov(\bfx, \bfy)+\cov(\bfy, \bfx))^2  + \bar{\bfy}^\top \Var[\bfx] \bar{\bfy} \notag\\
&\quad+ \bar{\bfx}^\top \Var[\bfy] \bar{\bfx } + \bar{\bfy}^\top \cov(\bfx, \bfy) \bar{\bfx} + \bar{\bfx}^\top \cov(\bfy, \bfx) \bar{\bfy}\\
&\begin{bmatrix}
    \cov(\bfx, \bfx^\top \bfy)
    \\
    \cov(\bfy, \bfx^\top \bfy)
    \\
    \cov(\bfz, \bfx^\top \bfy)
\end{bmatrix}
  = \begin{bmatrix}
  \cov(\bfx, \bfy)\bar{\bfx} + \Var[\bfx] \bar{\bfy}
  \\
  \Var[\bfy]\bar{\bfx} + \cov(\bfy, \bfx) \bar{\bfy}
  \\
  \cov(\bfz, \bfy)\bar{\bfx} + \cov(\bfz, \bfx) \bar{\bfy}
  \end{bmatrix}.
  \label{eq:dot-prod-gps-cov}
\end{align}
\end{lemma}

\begin{proof}
See Appendix~\ref{proof:dot-prod-of-gps}.
\end{proof}

\begin{lemma}[{\cite[Thm.~2.2.2]{adler1981geometry}\cite[p.~43]{taylorStats352Notes}}]
\label{lemma:differentiating-gp}
Let $q(\bfx)$ be a scalar Gaussian Process with differentiable mean function
$\mu(\bfx) : \R^n \mapsto \R$ and twice-differentiable covariance function $\knl(\bfx, \bfx') : \R^n \times \R^n \mapsto \R$. If $\grad_\bfx \mu(\bfx)$ exists and is finite for all $\bfx \in \bbR$ and $\calH_{\bfx, \bfx'}\knl(\bfx, \bfx') = [\frac{\p^2 \knl(\bfx, \bfx')}{\p \bfx_i, \p \bfx'_j} ]_{i=1,j=1}^{n,n}$ exists and is finite for all $(\bfx, \bfx') \in \bbR^{2n}$, then $q(\bfx)$ possesses a mean square derivative $\grad_\bfx q(\bfx)$, which is a vector-variate Gaussian Process $\GP(\grad_\bfx \mu(\bfx), \calH_{\bfx, \bfx'}\knl(\bfx, \bfx'))$. If $s$ is another random process whose finite covariance with $q$ is $\cov_{q,s}(\bfx, \bfx')$, then
\begin{equation}
\label{eq:differentiating-gp-cov}
\cov(\grad_\bfx q(\bfx), s(\bfx')) = \grad_\bfx \cov_{q,s}(\bfx, \bfx').
\end{equation}
\end{lemma}



The mean and variance of $\CBCtwo(\bfx, \bfu)$ can be computed using Lemma~\ref{prob:VarUquadrac24562}:
\begin{align}
\label{eq:mean-var-cbc}
&\mCBCtwo = \bfe^{(2)}(\bfx)^\top\ubfu
\notag\\
&\VCBCtwo = \ubfu^\top \bfV^{(2)}(\bfx)\bfV^{(2)}(\bfx)^\top\ubfu
\notag\\
&\bfe^{(2)}(\bfx) = \E[\dynAff(\bfx)^\top \grad_\bfx \LieOne h(\bfx)  ]
\notag\\
  &\qquad \qquad
  + \begin{bmatrix}\bfk_{\alpha,1} h(\bfx) + \bfk_{\alpha,2} \E[\LieOne h(\bfx)] \\ \mathbf{0}_m \end{bmatrix}
\notag\\
&\bfV^{(2)}(\bfx) = \biggl( \Var[ \dynAff(\bfx)^\top \grad_\bfx \LieOne h(\bfx) ] \\
&\qquad\qquad
+ \begin{bmatrix}
\bfk_{\alpha,2}^2 \Var[ \LieOne h(\bfx) ]  + 2[\bfd(\bfx)]_{1}& [\bfd(\bfx)]_{2:m+1}^\top
\notag\\ [\bfd(\bfx)]_{2:m+1} & \mathbf{0}_{m \times m} \end{bmatrix} \biggr)^\frac{1}{2}
\notag\\
&\bfd(\bfx) \triangleq \bfk_{\alpha,2} \cov( F(\bfx)^\top \grad_\bfx \LieOne h(\bfx) , \LieOne h(\bfx)) \notag
\end{align}
where the means, variances, and covariances of $\Lie_{f}h(\bfx)$ and $F(\bfx)^\top\grad_\bfx \Lie_{f}h(\bfx)$ in the above equations can be computed via Algorithm~\ref{alg:Algo} using Lemma~\ref{lemma:dot-prod-gps} and Lemma~\ref{lemma:differentiating-gp}.

%

\begin{algorithm}[t]
  \DontPrintSemicolon
  \LinesNumbered
  \SetAlgoLined
  \KwData{ECBF $h(\bfx)$, system dynamics distribution $\vect(F(\bfx)) \sim \GP(\vect(\bfM_k(\bfx)), \bfB_k(\bfx, \bfx') \otimes \bfA)$}
  \KwResult{$\mCBCtwo$ and $\VCBCtwo$}
  Compute $\E[\LieOne h(\bfx)] = \grad_\bfx h(\bfx)^\top \bfM_k(\bfx)[1, \mathbf{0}_m^\top]^\top$ and $\Var[\LieOne h(\bfx)] = \brl{\bfB_k(\bfx, \bfx')}_{1,1} (\grad_\bfx h(\bfx)^\top \bfA \grad_\bfx h(\bfx))$ by substituting $\ubfu = [1, \mathbf{0}_m^\top]^\top$ in Lemma~\ref{lem:NRcalculation}.
  \;
  Compute the mean and variance of $\grad_\bfx \LieOne h(\bfx)$ using Lemma~\ref{lemma:differentiating-gp},
  \;
  Compute $\cov\bigl(\grad_\bfx \LieOne h(\bfx),\LieOne h(\bfx)\bigr)$ using \eqref{eq:differentiating-gp-cov} 
  \;
  Compute the mean and variance $\grad_\bfx[ \LieOne h(\bfx)]^\top F(\bfx)\ubfu$ using
  Lemma~\ref{lemma:dot-prod-gps},
  \;
  Compute~\mbox{$\bfd(\bfx)^\top \ubfu = \cov\bigl(\grad_\bfx[ \LieOne h(\bfx)]^\top F(\bfx)\ubfu, \LieOne h(\bfx)\bigr)$} using \eqref{eq:dot-prod-gps-cov},
  \;
  Plug the above values into \eqref{eq:mean-var-cbc} to obtain $\mCBCtwo$ and $\VCBCtwo$.
  \caption{Mean and variance of $\CBCtwo(\bfx,\bfu)$.}
\label{alg:Algo}
\end{algorithm}

%% file: tex/mean-and-variance-of-cbf-r-diagram.tex
\begin{tikzpicture}
\providecommand{\bfx}{\mathbf{x}}
\providecommand{\Lie}{\mathcal{L}}
\providecommand{\grad}{\nabla}
\node [] (gh) {$\grad_\bfx h(\bfx)$};

\node [right=of gh] (L1) {$\underbrace{\grad_\bfx h(\bfx)^\top f(\bfx)}_{\Lie_f h(\bfx)}$};

\path [->,color=red] (gh)  edge [above] node {$[\cdot]$} (L1);

\node [below=of L1] (f1) {$f(\bfx) \sim \GP $};

\path [->,color=red] (f1)  edge [left] node {$[\cdot]$} (L1);

\node [right=of L1] (gL1) {$\grad_\bfx \Lie_f h(\bfx)$};

\path [->,color=blue] (L1)  edge [above] node {$\grad_\bfx$} (gL1);

\node [right=of gL1] (L2) {$\grad_\bfx[\Lie_f h(\bfx)]^\top F(\bfx) \ubfu$};

\node [below=38pt of L2] (f2) {$F(\bfx)\ubfu \sim \GP $};

\path [->,color=red] 
(f2) edge 
    node[above] {$[\cdot]$}
    node[below] {$\ubfu = [1, \mathbf{0}^\top_m]^\top$} 
 (f1);

\path [->,color=red] (f2)  edge [left] node {$[\cdot]$} (L2);

\path [->,color=red] (gL1)  edge [above] node {$[\cdot]$} (L2);

\end{tikzpicture}

%% file: tex/Evaluation.tex
\section{Simulations}

We evaluate our proposed MVGP learning model and
the SOCP-based safe control synthesis on two simulated systems: (A) Pendulum and (B) Ackermann Vehicle. To allow reproducing the results, our implementation is available on Github\footnote{\url{https://github.com/wecacuee/Bayesian_CBF}}.


\subsection{Pendulum}
\label{sec:eg-pendulum}

Consider a pendulum, shown in Fig~\ref{fig:CPS}, with state $\bfx = [\theta, \omega]^\top$, containing its angular deviation $\theta$ from the rest position and its angular velocity $\omega$. The pendulum dynamics are:
\begin{align}
    \begin{bmatrix}
    \dot{\theta} \\ \dot{\omega}\end{bmatrix}
    = \underbrace{\begin{bmatrix}
    \omega \\ -\frac{g}{l}\sin(\theta)
    \end{bmatrix}}_{f(\bfx)}
    +
    \underbrace{
    \begin{bmatrix}
    0\\
    \frac{1}{ml}
    \end{bmatrix}}_{g(\bfx)}
    u,
    \label{eq:pendulum}
\end{align}
where $g$ is the gravitational acceleration, $m$ is the mass, and $l$ is the length.

\subsubsection{Estimating Pendulum Dynamics}
\label{sec:pendulum-learning}
We compare our MVGP model versus the Coregionalization GP (CoGP) \cite{alvarez2012kernels} in estimating the pendulum dynamics using data from randomly generated control inputs. In our simulation, the true pendulum model has mass $m = 1$ and length $l = 1$. The inference algorithms were implemented in Python using GPyTorch~\cite{gardner2018gpytorch}. We let $\bfK_0(\bfx, \bfx') = \bfB\knl_0(\bfx, \bfx') \otimes \bfA$ for MVGP and $\bfK_0(\bfx, \bfx') = \bfSigma\knl_0(\bfx, \bfx')$ for CoGP, where $\knl_0(\bfx, \bfx')$ is a scaled radial-basis function kernel. Further, $\bfA$, $\bfB$, and $\bfSigma$ are constrained to be positive semi-definite by modeling each as $\bfC_r\bfC_r^\top + \diag(\bfv)$, where $\bfC_r \in \bbR^{l \times r}$, $r$ is a desired rank and $\diag(\bfv)$ is a diagonal matrix constructed from a vector $\bfv$.


We emphasize that a straightforward implementation of decoupled GPs for each system dimension is not possible because the training data is only available as a linear combination of the unknown functions $f(\bfx)$ and $g(\bfx)$. We approximate decoupled GP inference by constraining the matrices $\bfA$, $\bfB$, and $\bfSigma$ in the MVGP and CoGP models to be diagonal. To compare the accuracy of the GP models, we randomly split the samples from the pendulum simulation into training data and test data. Given a test set $\calT$, we compute \emph{variance-weighted error} as,
\begin{align}
    \mbox{err}(\calT) \triangleq \sqrt{
    \sum_{\bfx \in \calT} 
    \frac{
    \| \bfK_k^{-\frac{1}{2}}(\bfx, \bfx)\vect(\bfM_k(\bfx) - F(\bfx))\|_2^2
    }{|\calT|}
    }.
    \label{eq:variance-weighted-error}
\end{align}
A qualitative comparison of MVGP and CoGP is shown in Fig.~\ref{fig:learning-accuracy}. A quantitative comparison of the computational  efficiency  and  inference  accuracy of the models with \emph{full} and \emph{diagonal} covariance matrices is shown in Fig.~\ref{fig:computation-mvgp-vs-coregionalization}. The median variance-weighted error and error bars showing the 2nd to 9th decile over 30 repetitions of the evaluation are shown. The experiments are performed on a desktop with Nvidia\textsuperscript{\textregistered} GeForce RTX\textsuperscript{TM} 2080Ti GPU and Intel\textsuperscript{\textregistered} Core\textsuperscript{TM} i9-7900X CPU (3.30GHz). The results show that MVGP inference is significantly faster than CoGP, while maintaining comparable accuracy. Both MVGP and CoGP have higher median accuracy than their counterparts with diagonally restricted covariances.



%

\begin{figure*}
\includegraphics[width=\linewidth, trim=0 0 0 10pt, clip]{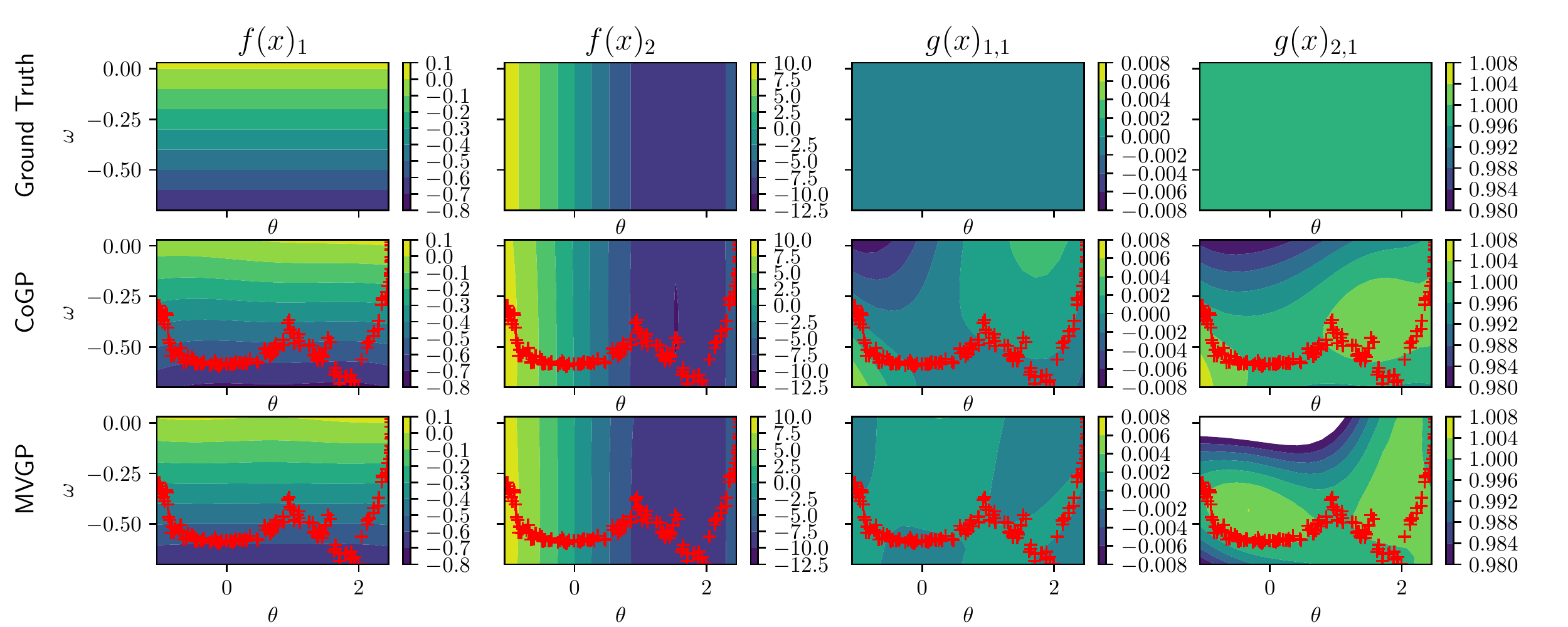}
\caption{Qualitative comparison of our Matrix Variate Gaussian Process (MVGP) regression with the Coregionalization GP (CoGP) model \cite{alvarez2012kernels}. Training sample are generated using random control inputs to the pendulum. The training samples are shown as {\color{red}+}. The learned models are evaluated on a $20 \times 20$ grid, shown as contour plots. The MVGP model has lower computational complexity than CoGP ($O(k^3)$ vs $O(k^3n^3)$, see Sec.~\ref{learningsec:13452}) without significant drop in accuracy.}
\label{fig:learning-accuracy}
\end{figure*}

\begin{figure}
  \begin{center}
    \hspace{14pt}\includegraphics[width=0.90\linewidth]{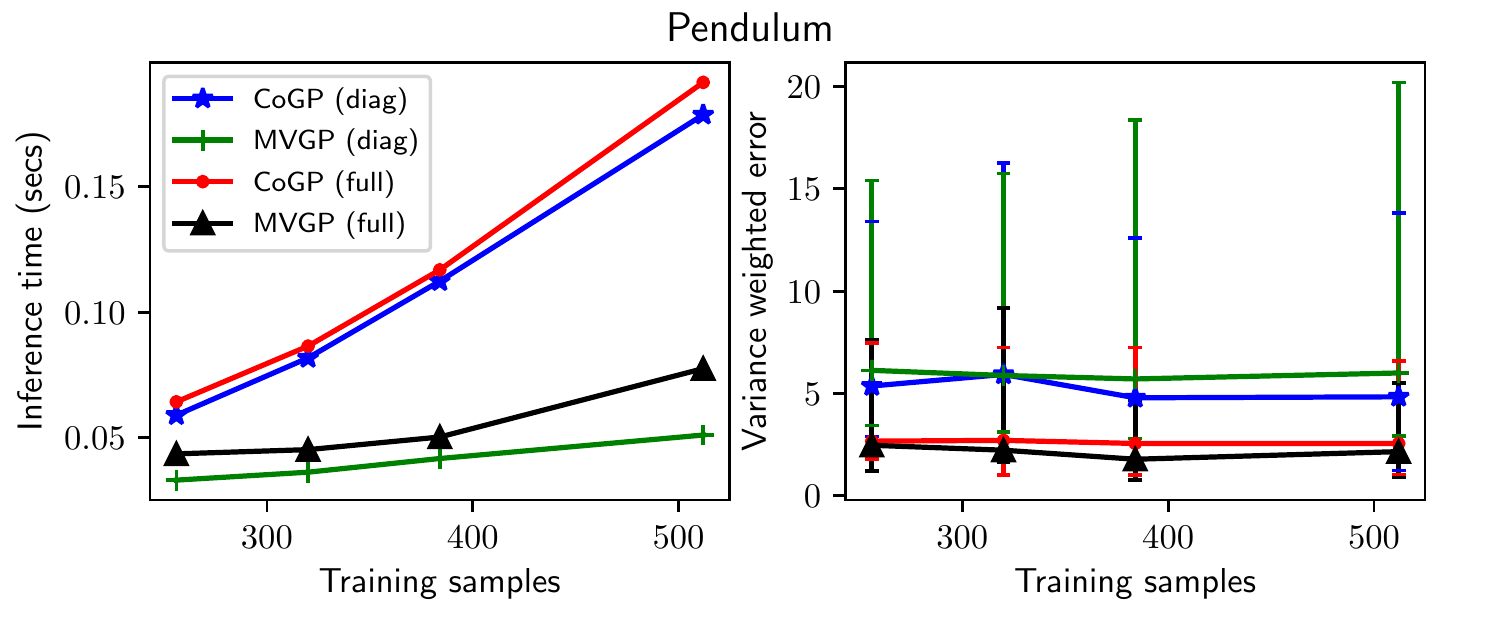}\\%
    \includegraphics[width=0.95\linewidth]{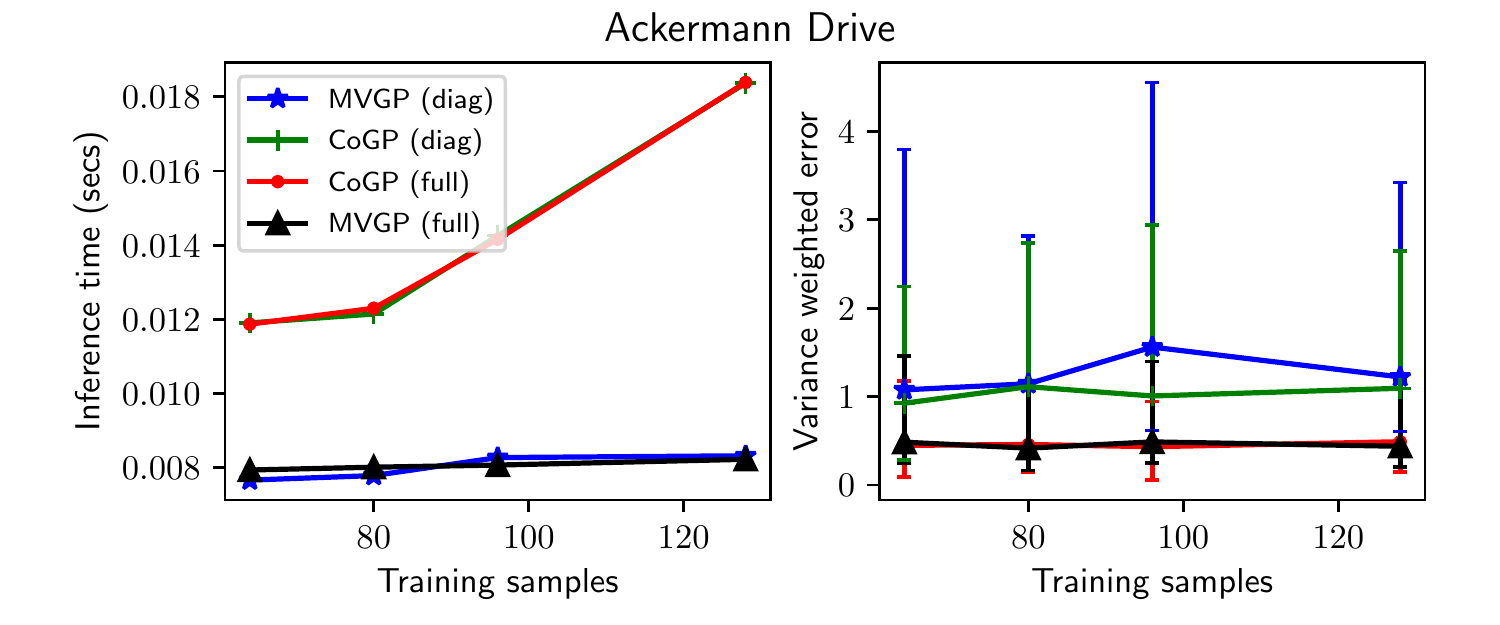}
    \end{center}
    \caption{Comparison of the computational efficiency and inference accuracy of MVGP and CoGP~\cite{alvarez2012kernels}.
    Both models are evaluated on the Pendulum and Ackermann systems in two modes: 1) with correlation matrix constrained to be diagonal (labelled \emph{diag}) and (2) without any constraint on correlation matrix (labelled \emph{full}).
    The results show that MVGP is significantly faster than CoGP. We argue in Sec.~\ref{learningsec:13452} that the computational complexity of MVGP is $O(k^3)$ for $k$ training examples, while that of CoGP is $O(k^3n^3)$, where $n$ is the state dimension.
    The \emph{variance-weighted error} is computed via \eqref{eq:variance-weighted-error}.
    The error bars denote the range from the 2nd to the 9th decile over 30 repetitions, centered around the median.
    The median errors of the MVGP and CoGP models are similar but enforcing diagonal covariance matrices, increases the median error.
    }
    \label{fig:computation-mvgp-vs-coregionalization}
    \label{fig:Ackermann-speed-test}
\end{figure}

\subsubsection{Safe Control of Learned Pendulum Dynamics}
A safe set is chosen as the complement of a radial region $[\theta_c -\Delta_{col}, \theta_c + \Delta_{col}]$ with $\theta_c = 45^\circ$, $\Delta_{col} = 22.5^\circ$ that needs to be avoided, as shown in Fig.~\ref{fig:CPS}. The control barrier function defining this safe set is $h(\bfx) = \cos(\Delta_{col}) - \cos(\theta - \theta_c)$. The controller knows a priori that the system is control-affine with relative degree two, but it is not aware of $f$ and $g$. A zero-mean prior $\bfM_0(\bfx) = \mathbf{0}_{n \times (1+m)}$ and randomly generated covariance matrices $\bfA$ and $\bfB$ are used to initialize the MVGP model. We formulate a SOCP as in \eqref{progg!!5544i} with $r=2$. We specify a task requiring the pendulum to track a reference control signal $\hat{\pi}_k(\bfx)$ and specify the optimization objective as $\|\bfR(\bfx_k)(\bfu_k-\hat{\pi}_k(\bfx_k))\|_2^2$ with $\bfR(\bfx_k) \equiv \bfI$. The reference signal $\hat{\pi}_k(\bfx)$ is an $\epsilon$-greedy policy \cite{sutton2018reinforcement}, used to provide sufficient excitation in the training data. Concretely, $\hat{\pi}_k(\bfx)$ is sampled from an $\epsilon_k$ weighted combination of Direc delta distribution $\delta$ and Uniform distribution $\bbU$, $\hat{\pi}(\bfx_k) \sim (1-\epsilon_k) \delta_{\bfu = 0} + \epsilon_k \bbU[-20, 20]$, where $\epsilon_k = 10^{-k/50}$ so that $\epsilon_k$ goes from $1$ to $0.01$ in $100$ steps. We initialize the system with parameters $\theta_0=75^\circ$, $\omega_0=-0.01$, $\tau=0.01$, $m=1$, $g=10$, $l=1$. Our simulation results show that the pendulum remains in the safe region (see Fig.~\ref{fig:CPS}). Negative control inputs get rejected by the SOCP safety constraint, while positive inputs allow the pendulum to bounce back from the unsafe region.


\begin{figure}
\centering
\includegraphics[width=0.30\linewidth]{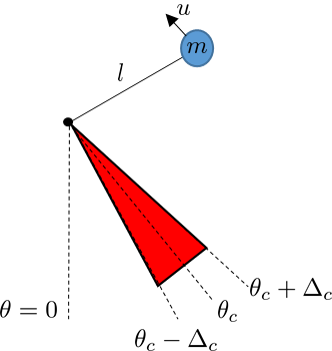}%
\includegraphics[width=0.70\linewidth]{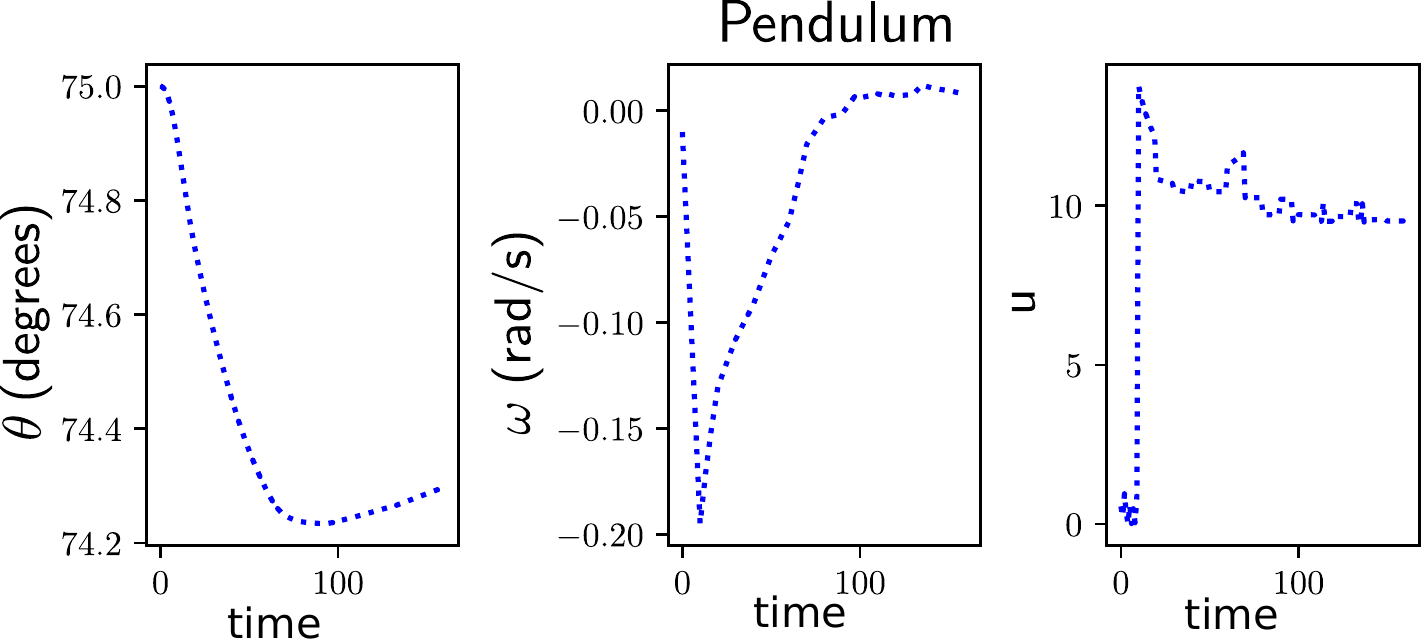}\\
\caption{Pendulum simulation (left) with an unsafe (red) region. The pendulum trajectory (middle two) resulting from the application of safe control inputs (right) is shown. The pendulum starts from $\theta = 75^\circ$, drops down until it reaches the unsafe region and then stays there.}
\label{fig:CPS}
\end{figure}

\subsection{Ackermann Vehicle}
\label{sec:sims-Ackermann}

Consider an Ackermann-drive vehicle model with state $\bfx = [x, y, \theta]^\top$, including position $(x,y)$ and orientation $\theta$. To make the Ackermann dynamics affine in the control input, we define $z \triangleq v\tan(\phi)$, where $v$ is the linear velocity and $\phi$ is the steering angle. The Ackermann-drive dynamics are:
\begin{align}
    \begin{bmatrix}
    \dot{x} \\ \dot{y} \\ \dot{\theta}
    \end{bmatrix}
    = \underbrace{\begin{bmatrix}
    0 \\ 0 \\ 0
    \end{bmatrix}}_{f(\bfx)}
    +
    \underbrace{
    \begin{bmatrix}
    \cos(\theta) & 0\\
    \sin(\theta) & 0 \\
    0 & \frac{1}{L} \\
    \end{bmatrix}}_{g(\bfx)}
    \underbrace{
    \begin{bmatrix}
    v \\
    z
    \end{bmatrix}}_{\bfu},
    \label{eq:ctrl-aff-ackermann}
\end{align}
where $L$ is the distance between the front and back wheels.


\subsubsection{Estimating Ackermann-drive Dynamics}
Similar to Sec.~\ref{sec:pendulum-learning}, we compare the computational efficiency and inference accuracy of MVGP and CoGP with \emph{diagonal} and \emph{full} covariance matrices on training data generated from the Ackermann-drive model in \eqref{eq:ctrl-aff-ackermann}. We further explicitly assume translation-invariant dynamics, i.e, $F(\bfx) = F([0, 0, \theta]^\top)$, $\forall \bfx \in \bbR^3$. This assumption allows us to transfer the learned dynamics to unvisited parts of the environment. The results of the simulation are shown in Fig.~\ref{fig:Ackermann-speed-test}. We again find that MVGP inference is significantly faster than CoGP. In terms of accuracy, MVGP is comparable to CoGP and restricting the covariance matrices to diagonal matrices increases the median variance-weighted error error.

\subsubsection{Safe Control of Learned Ackermann-drive Dynamics}
To test our safe controller on the Ackermann-drive vehicle, we consider navigation to a goal state in the presence of two circular obstacles in the environment. The CBF for circular obstacle $i \in \{1, 2\}$ with center $\bfo_i \in \bbR^2$ and radius $r_i > 0$ is:
\begin{align}
h_i(\bfx) &= q_1 (\|\bfx_{1:2} - \bfo_i\|^2_2 - r_i^2) + q_2\cos(\theta - \phi_o),
\end{align}
where $\phi_o = \tan^{-1}\left(\frac{y - \bfo_{i,2}}{x - \bfo_{i,1}}\right)$. We assume that a planning algorithm provides a desired time-parameterized trajectory $\bfx^{(d)} = [x^{(d)}, y^{(d)}, \theta^{(d)}]$ and its time derivative $\dot{\bfx}^{(d)}$.
We select a control Lyapunov function candidate:
\begin{equation}
V(\bfx, \bfx^{(d)}) = \frac{w_1}{2}\rho^2 + w_2(1-\cos(\alpha)),
\end{equation}
where $\rho \triangleq \|\bfx_{1:2}^{(d)} - \bfx_{1:2}\|_2$ and $\alpha \triangleq \theta - \tan^{-1}\left(\frac{y^{(d)} - y}{x^{(d)}-x}\right)$. The control Lyapunov condition is given by:
%
\begin{align}
    \mbox{CLC}(\bfx, \bfu) &=
    \grad_\bfx^\top V(\bfx, \bfx^{(d)})  F(\bfx) \ubfu +
    \notag\\
    &\quad \grad_{\bfx^{(d)}}^\top V(\bfx, \bfx^{(d)}) \dot{\bfx}^{(d)}
    + \gamma_v V(\bfx, \bfx^{(d)}).
\end{align}
%
The control Barrier condition for each obstacle $i$ is:
\begin{align}
    \mbox{CBC}(\bfx, \bfu; i) = \grad_\bfx^\top h_i(\bfx) F(\bfx)\ubfu + \gamma_{h_i} h_i(\bfx).
\end{align}
%
We implement
the SOCP controller in Prop.~\ref{prop:socp-rel-deg-1} with parameters: $q_1 = 0.7$, $q_2 = 0.3$, $w_1 = 0.9$, $w_2 = 1.5$, $\gamma_{h_i} = 5$, $\gamma_V = 10$, $\tau = 0.01$.

We perform two experiments to evaluate the advantages of our model.
First, we evaluate the importance of \emph{accounting for the variance} of the dynamics estimate for safe control.
Second, we evaluate the importance of \emph{online learning} in reducing the variance and ensuring safety.

Recall the observation in Prop.~\ref{prop:socp-rel-deg-1} that $c(\tilde{p}_k = 0.5) = 0$. In this case, the SOCP in \eqref{lower:nu1} reduces to the deterministic-case QP in~\eqref{eq:clf-cbf-qp}. We dub the case of $\tilde{p}_k = 0.5$ as \emph{Mean CBF} and the case of $\tilde{p}_k = 0.99$ as \emph{Bayes CBF}. In both the cases, the covariance matrices are initialized as $\bfA = 0.01\bfI_{3}$ and $\bfB = \bfI_3$. The true dynamics, $F(\bfx)$, is specified with $L=12$, while the mean dynamics $\bfM_0(\bfx)$ is obtained with $L=1$. The result of simulation is shown in Fig.~\ref{fig:ex-Ackermann-mean-vs-bayes}. The \emph{Bayes CBF} controller slows down when the safety constraint in \eqref{eq:stochastic-cbc-rd-1} hits zero and then changes direction while avoiding the obstacle. However, the \emph{Mean CBF} controller is not able to avoid the obstacle because it does not take the variance of the dynamics estimate into account and the mean dynamics are inaccurate.

We also evaluate the importance of \emph{online learning} for safe control. We compare two setups: (1) updating the system dynamics estimate online via our MVGP approach every 40 steps using the data collected so far (\emph{With Learning}) and (2) relying on the prior system dynamics estimate only (\emph{No Learning}). In both the cases, we choose $\tilde{p}_k = 0.99$. The true dynamics are specified with $L=1$, while prior mean dynamics correspond to $L=8$. The covariance matrices $\bfA$ and $\bfB$ are initialized with elements independently sampled from a standard Gaussian distribution while ensuring them to positive semi-definite as before. The result are shown in Fig.~\ref{fig:ex-Ackermann-learning}. Online learning of the vehicle dynamics reduces the variance and allows the vehicle to pass between the two obstacles. This is not possible using the prior dynamics distribution because the large variance makes the safety constraint in \eqref{eq:stochastic-cbc-rd-1} conservative.


\begin{figure}
\includegraphics[width=0.527\linewidth]{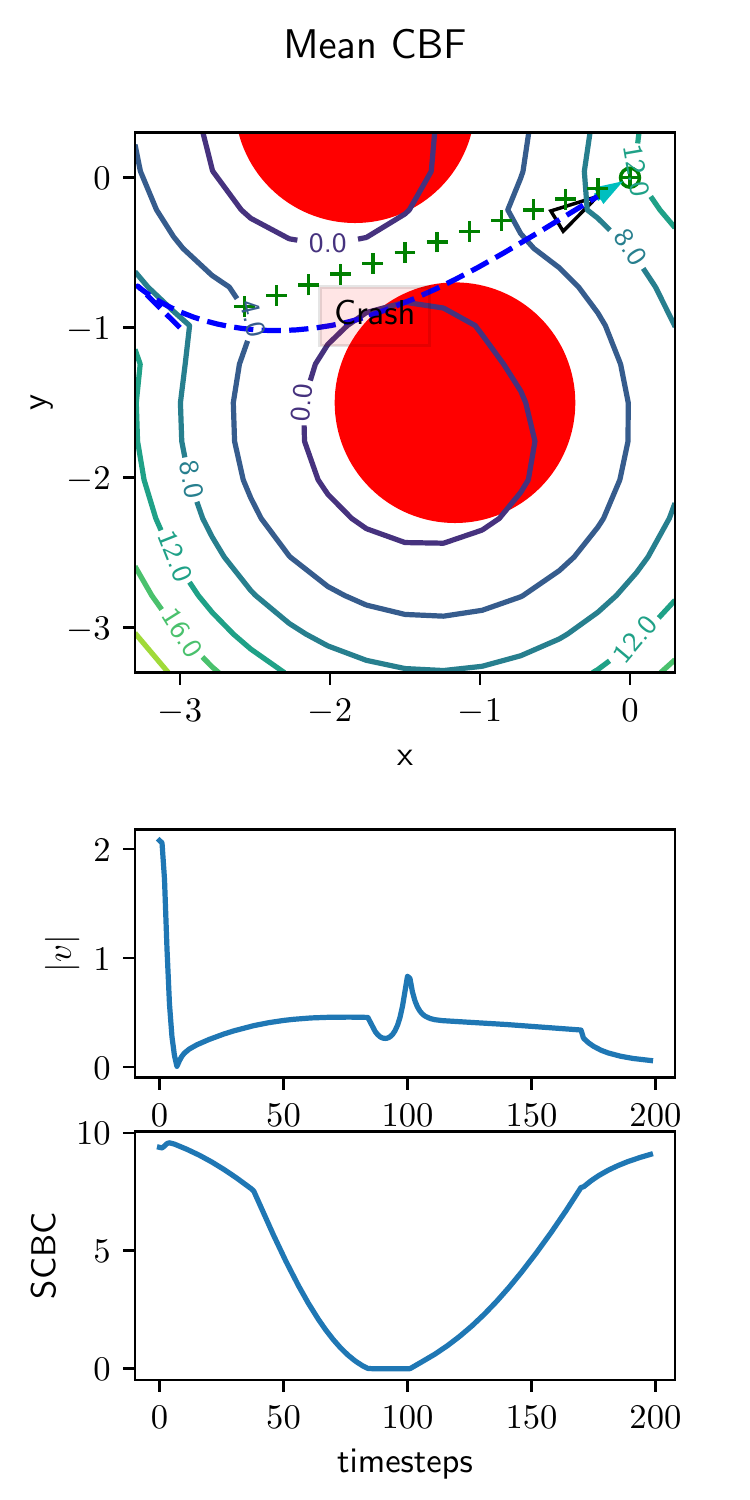}%
\includegraphics[width=0.473\linewidth,trim=22pt 0pt 0pt 0pt, clip]{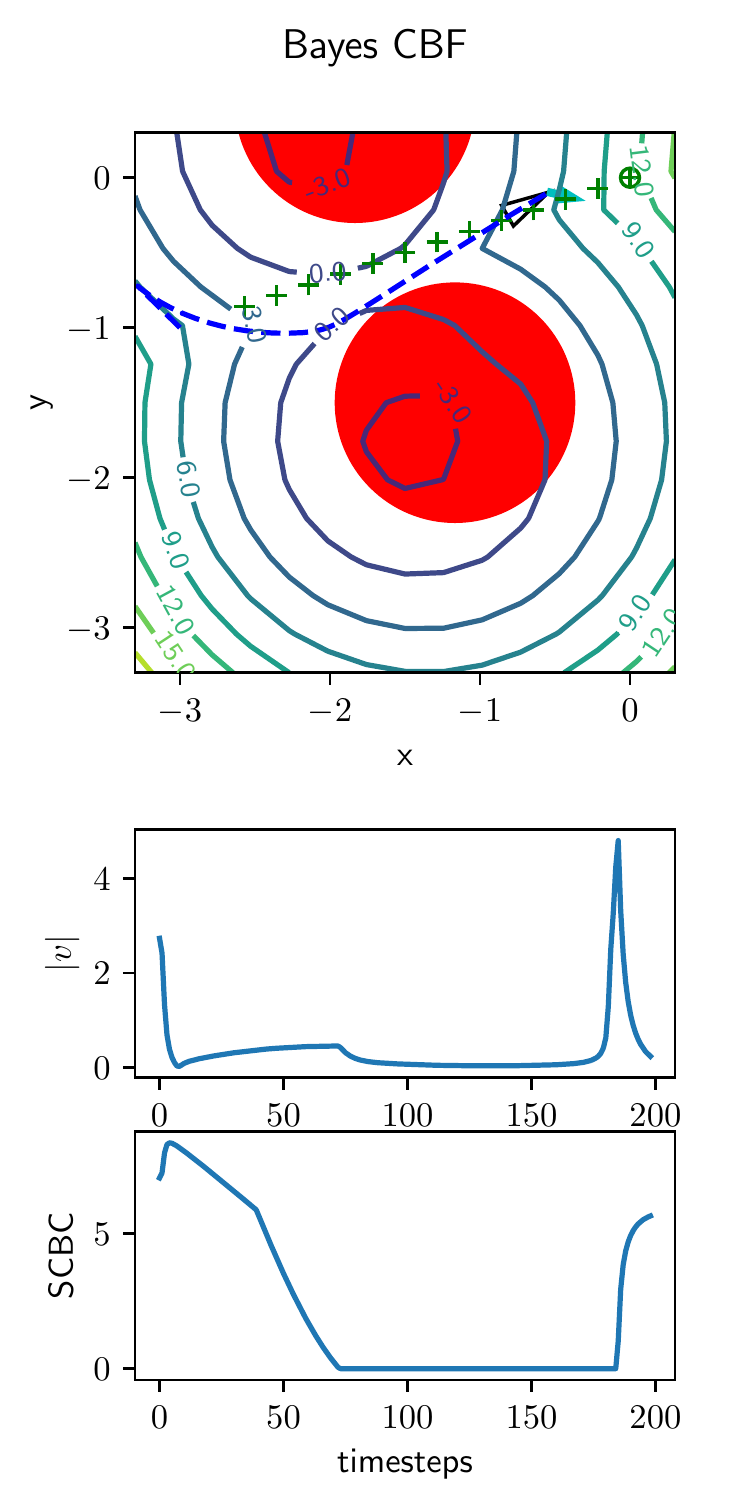}
\caption{
Comparison of enforcing CBF constraints with Ackermann dynamics when accounting (\emph{Bayes CBF}) and not accounting (\emph{Mean CBF}) for the variance in the dynamics estimate.
The top row shows the Ackermann vehicle trajectory in dashed blue with two obstacles in red.
The contour plots shows the minimum of the SCBC~\eqref{eq:stochastic-cbc-rd-1} values corresponding to the two obstacles, evaluated on the $(x,y)$ grid while keeping $\theta$ and $\bfu$ fixed. 
The middle row shows the magnitude of the velocity input over time.
The bottom row shows the minimum of the two SCBC~\eqref{eq:stochastic-cbc-rd-1} values over time. 
Enforcing safety using only the mean CBC (\emph{Mean CBF}) results in a collision, while accounting for stochastic CBC (\emph{Bayes CBF}) constraint causes the Ackermann vehicle to slow down and turn away from the unsafe region.
A video rendering of these simulations is available in the supplementary material.
}
\label{fig:ex-Ackermann-mean-vs-bayes}
\end{figure}


\begin{figure}
    \centering
    \includegraphics[width=0.521\linewidth]{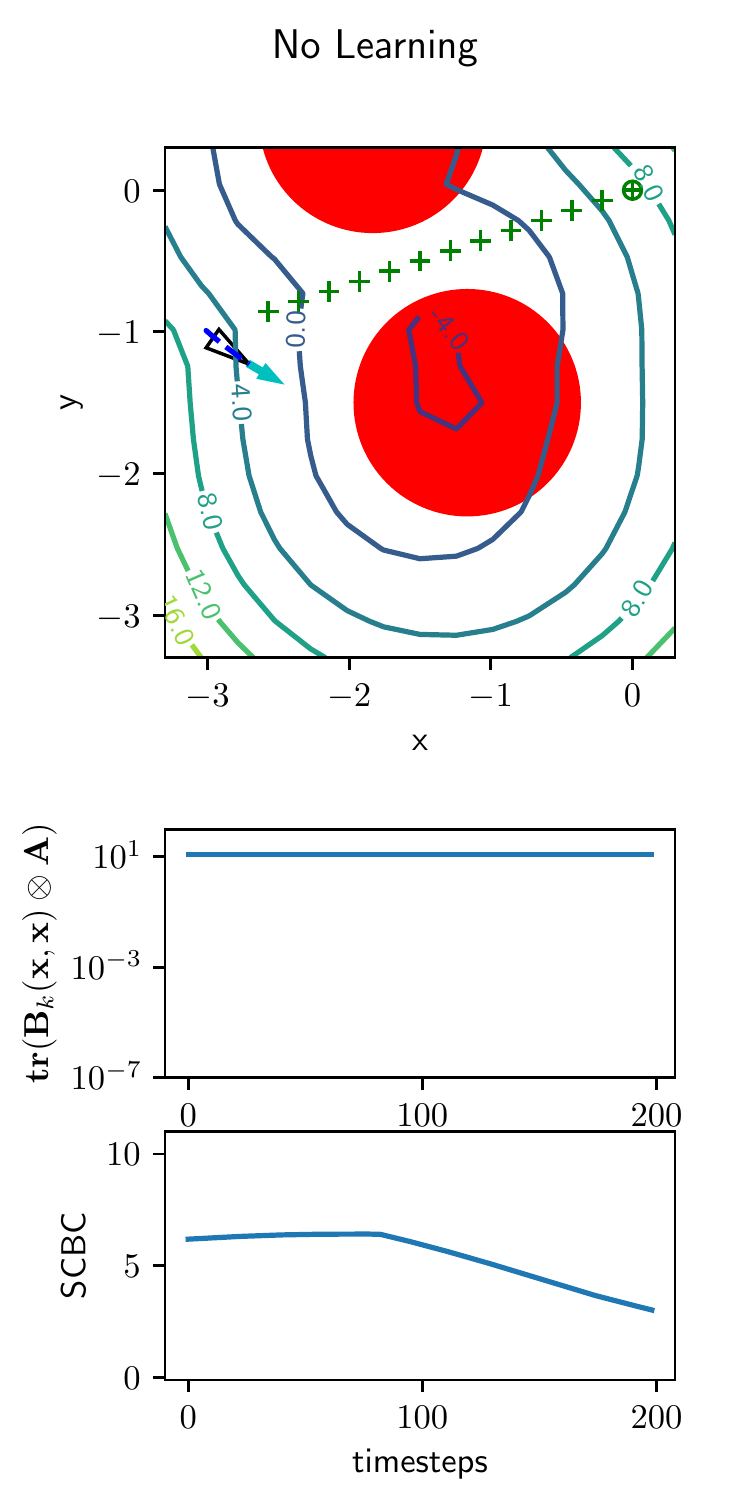}%
    \includegraphics[width=0.479\linewidth,trim=18pt 0pt 0pt 0pt,clip]{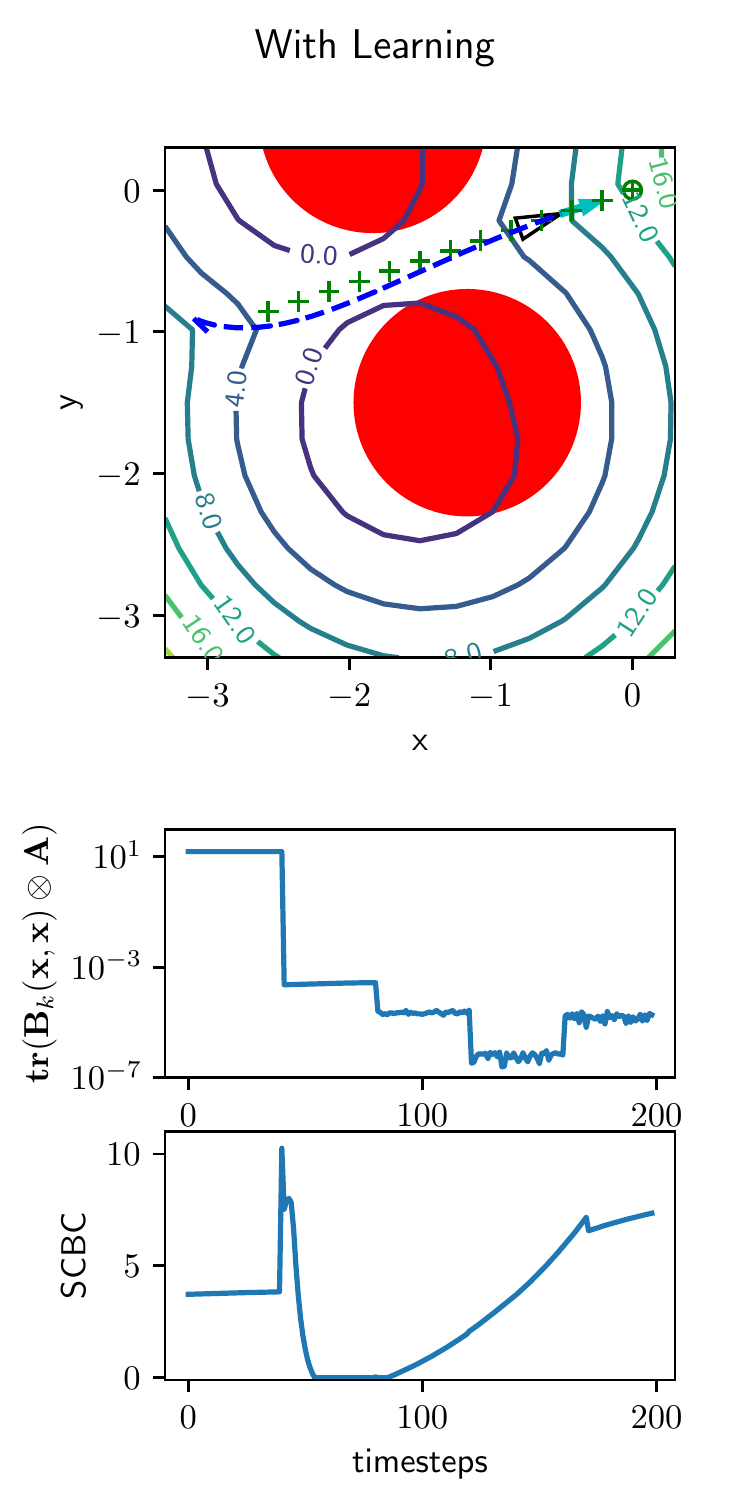}
    \caption{The effect of online dynamics learning (right) versus no online learning (left) on the safe control of an Ackermann vehicle. The top row shows the vehicle trajectory in dashed blue with two obstacles in red. The middle row shows the trace of the covariance matrix $\tr(\bfB_k(\bfx, \bfx) \otimes \bfA)$, which we use as a measure of uncertainty. The bottom row shows the minimum of the two probabilistic safety constraint over time, as defined in \eqref{eq:stochastic-cbc-rd-1}. Note that without learning, the vehicle gets stuck between the two obstacles because the uncertainty in the dynamics is too high, i.e., the safety condition in \eqref{eq:stochastic-cbc-rd-1} cannot be rendered positive. With online learning, however, the uncertainty is reduced enough to allow the safety condition to become positive in the area between the two obstacles. The dynamics distribution is updated every 40 time steps. Note the drop in uncertainty in the middle row at these time steps. A video rendering of these simulations is available in the supplementary material.}
    \label{fig:ex-Ackermann-learning}
\end{figure}

\Revised{}{
\subsubsection{Evaluating triggering time}
\label{sec:self-triggering-design-sim}
We evaluate the bounds on the Lipschitz constant $L_{\bff_k}$ (Lemma~\ref{lem:LipschitzGP}) and triggering time $\tau_k$ (Prop~\ref{thm:validity-inter-triggeringt-times}) for the Ackermann vehicle. At any time $k$, we use the following choice of parameters to compute $\tau_k$, $\delta_L=10^{-4}$, $\zeta = 10^{-2}$, $L_\alpha = 1$, and $L_{h_k} = \max_{\bfx \in \calX_k} \grad_\bfx h(\bfx)$. We define $\calX_k$ to be a cuboid with sides $(0.2, 0.2, \frac{\pi}{50})$ in the directions $x, y$ and $\theta$ respectively centered around the current state $\bfx_k$. We compute the Lipschitz constant $L_{\bff_k}$ and $\tau_k$ both numerically and analytically. The numerical computation of $L_{i, \partial j}$ is done by sampling the GP, $\frac{\partial \bff_i(\bfx)}{\partial \bfx_j}$ over a $10\times10\times10$ grid in $\calX_k$ and taking the maximum norm as the Lipschitz constant. The rest of the computation for the numerical approximation of $\tau_k$ is the same as the analytical $\tau_k$. The results are shown in Fig.~\ref{fig:triggering-time}. The shapes of the numerical and analytical Lipschitz constants are the same, although they differ by a couple orders of magnitude in scale since the analytical bound is conservative. Also, note that we are that in the numerical approximation $\delta_L$ is determined by the samples drawn from the GP. We also observe that the second term in the expression for $L_{i, \partial j}$ (Lemma~\ref{lem:LipschitzGP}) dominates the effect of $\delta_L$ on the Lipschitz constant. 
}

%% file: tex/conculsion.tex
\section{Conclusion}\label{sec:conc}
Allowing artificial systems to safely adapt their own models during online operation will have significant implications for their successful use in unstructured changing real-world environments. This paper developed a Bayesian inference approach to approximate system dynamics and their uncertainty from online observations. The posterior distribution over the dynamics was used to specify probabilistic constraints that guarantee safe and stable online operation with high probability. Our results offer a promising approach for controlling complex systems in challenging environments. Future work will focus on applications of the proposed approach to real robot systems.


%% file: tex/MGVresults.tex

\section{Additional properties of the MVG distribution}\label{sec:MGVresutls2}

\begin{lemma}[{\cite{StructuredPBP}}]
\label{lemma:mvg2gaussian}
Let $\bfX$ follow an MVG distribution $\calM\calN(\bfM,\bfA,\bfB)$. Then, $\vect(\bfX) \sim \cal{N}\prl{ \textit{vec}(\bfM), \bfB\otimes\bfA}$.
\end{lemma}

\begin{lemma}[\cite{StructuredPBP}]
\label{lemma:product-mvg}
Let $\bfX$ follow an MVG distribution $\calM\calN(\bfM,\bfA,\bfB)$ and let $\bfC \in \mathbb{R}^{d \times n}$ and $\bfD \in \mathbb{R}^{m \times d}$. Then,
\begin{equation}
\label{sunerrww233!}
\begin{aligned}
  \bfC \bfX &\sim \calM\calN(
  \bfC\bfM, \bfC\bfA\bfC^\top,\bfB),
  \\ 
\bfX \bfD &\sim \calM\calN(\bfM\bfD, \bfA,\bfD^\top\bfB\bfD).
\end{aligned}
\end{equation}
\end{lemma}

%% file: tex/mean-and-variance-of-cbf-2/matrix-variate-gp.tex
\section{Proof of Proposition~\ref{thmfrmvg24}}
\label{proof:matrix-variate-gp}

\newcommand{\ubcalB}{\underline{\bcalB}}

Let $\GP(\vect(\bfM_k(\bfx)), \bfK_k(\bfx, \bfx'))$ be the posterior distribution of $\vect(F(\bfx))$ conditioned on the training data $(\StDat_{1:k}, \CtDat, \StDtDat_{1:k})$. The mean and variance can be obtained by applying Schur's complement to \eqref{eq:joint-training-test-distribution},
\begin{align}
    &\vect(\bfM_k(\bfx)) = \vect(\bfM_0(\bfx)) +
    \\\notag
    &\qquad \bigl(  
    \ubcalB_{1:k}(\bfx) \otimes \bfA \bigr)^\top\bigl(\ubcalB_{1:k}^{1:k} \otimes \bfA \bigr)^{-1} \vect(\StDtDat - \bcalM_{1:k}\CtDat),
\\
&\bfK_k(\bfx, \bfx) = 
\bfB_0(\bfx, \bfx) \otimes \bfA -
\\\notag
&\qquad \bigl( \ubcalB_{1:k}(\bfx) \otimes \bfA \bigr)^\top
\bigl( \ubcalB_{1:k}^{1:k}  \otimes \bfA \bigr)^{-1}
\bigl( \ubcalB^\top_{1:k}(\bfx) \otimes \bfA \bigr).
\end{align}
where 
$\ubcalB_{1:k}^{1:k} \triangleq \CtDat^\top \bcalB_{1:k}^{1:k} \CtDat + \sigma^2 \bfI_k$ and
$\ubcalB_{1:k}(\bfx) \triangleq  \bcalB_{1:k}(\bfx)\CtDat$. For appropriately sized matrices $\bfP$, $\bfQ$, $\bfR$, $\bfS$, the Kronecker product satisfies $(\bfP \otimes \bfQ)(\bfR \otimes \bfS) = (\bfP\bfR \otimes \bfQ\bfS)$ and $(\bfP \otimes \bfQ)^{-1} = \bfP^{-1} \otimes \bfQ^{-1}$. Thus, we can rewrite the mean as:
\begin{equation}
\begin{aligned}
    &\vect(\bfM_k(\bfx)) = \vect(\bfM_0(\bfx)) + \\
    &\biggl(  \bigl(\ubcalB^\top_{1:k}(\bfx)[\ubcalB_{1:k}^{1:k}]^{-1}\bigr) \otimes \bfA\bfA^{-1} \biggr)
\vect(\StDtDat - \bcalM_{1:k}\CtDat).
\end{aligned}
\end{equation}
Applying $(\bfP \otimes \bfQ) \vect(\bfR) = \vect(\bfQ\bfR\bfP^\top)$, we get
\begin{align}
    \bfM_k&(\bfx) = \bfM_0(\bfx) + 
    (\StDtDat_{1:k} - \bcalM_{1:k}\CtDat)[\ubcalB_{1:k}^{1:k}]^{-1}
    \ubcalB^\top_{1:k}(\bfx)
    \notag\\
    &= \bfM_0(\bfx) +
    (\StDtDat_{1:k} - \bcalM_{1:k}\CtDat)(\CtDat\bcalB_{1:k}(\bfx))^{\dagger}.
\end{align}
Similarly, the covariance can be rewritten as,
\begin{align}
\bfK_k&(\bfx, \bfx') = \Bigl( \bfB_0(\bfx, \bfx') - \ubcalB_{1:k}(\bfx)[\ubcalB_{1:k}^{1:k} ]^{-1}\ubcalB^\top_{1:k}(\bfx')\Bigr)\otimes \bfA \notag\\
&=\Bigl(\bfB_0(\bfx, \bfx') - \bcalB_{1:k}(\bfx')\CtDat(\CtDat \bcalB_{1:k})^\dagger\Bigr)\otimes \bfA.
\end{align}
Defining $\bfB_k(\bfx, \bfx')$ such that $\bfK_k(\bfx, \bfx) = \bfB_k(\bfx, \bfx') \otimes \bfA$, we can write the posterior distribution of $\vect(F(\bfx))$ as $\mathcal{GP}(\vect(\bfM_k(\bfx)), \bfB_k(\bfx, \bfx') \otimes \bfA)$.\qed

\section{Coregionalization Gaussian Process}
\label{proof:coregionalization-gp}
\newcommand{\bcalK}{\boldsymbol{\calK}}
\newcommand{\ubcalK}{\underline{\bcalK}}


Here, we show how the Coregionalization model \cite{alvarez2012kernels} can be applied to Gaussian process inference when the training data is available as a matrix vector product.

\begin{lemma}
Let $\vect(F(\bfx)) \sim \GP(\vect(\bfM_0(\bfx)), \bfSigma\knl_0(\bfx, \bfx'))$, where $\bfSigma \in \bbR^{(1+m)n \times (1+m)n}$ is the covariance matrix of the output dimensions, and $\knl_0(\bfx, \bfx')$  is the kernel function. Denote the kernel function evaluation over the training data as,
\begin{equation}
\begin{aligned}
&\bcalK_{1:k}^{1:k} \triangleq \begin{bmatrix} 
\knl_0(\bfx_1, \bfx_1) & \ldots & \knl_0(\bfx_k, \bfx_1) \\
\vdots & & \vdots \\
\knl_0(\bfx_k, \bfx_1) & \ldots & \knl_0(\bfx_k, \bfx_k)
\end{bmatrix} \in \bbR^{k \times k}
\\
&\bcalK_{1:k}(\bfx) \triangleq \begin{bmatrix}
\knl_0(\bfx_1, \bfx) & \hdots & \knl_0(\bfx_k, \bfx)
\end{bmatrix}
\in \bbR^{1\times k}.
\end{aligned}
\end{equation}
Let $\bcalM_{1:k}$ and $\CtDat$ be defined as in Proposition~\ref{thmfrmvg24}. 
Then, the posterior distribution of $\vect(F(\bfx))$, given the training data $(\StDtDat_{1:k}, \StDat_{1:k}, \CtDat)$, is $\GP\bigl(\vect(\bfM_k(\bfx)), \bfK_k(\bfx, \bfx')\bigr)$, where
\begin{equation}
\label{eq:coregionalization-posterior}
\begin{aligned}
    &\vect(\bfM_k(\bfx)) \triangleq \vect(\bfM_0(\bfx)) +
    \\
     &\qquad\quad
    \ubcalK_{1:k}(\bfx)
    [\ubcalK_{1:k}^{1:k}]^{-1}
     \vect(\StDtDat_{1:k} - \bcalM_{1:k} \CtDat),
    \\
    &\bfK_k(\bfx, \bfx') \triangleq \bfSigma \knl_0(\bfx, \bfx') 
    -  
    \ubcalK_{1:k}(\bfx)
    [\ubcalK_{1:k}^{1:k}]^{-1}
    \ubcalK^\top_{1:k}(\bfx),
\end{aligned}
\end{equation}
$\ubcalK_{1:k}^{1:k} \triangleq (\CtDat^\top \otimes \bfI_n) 
    (\bcalK_{1:k}^{1:k}  \otimes \bfSigma)
    (\CtDat \otimes \bfI_n) + \bfI_k \otimes \bfS$, and $\ubcalK_{1:k}(\bfx) \triangleq (\bcalK_{1:k}(\bfx) \otimes \bfSigma)(\CtDat \otimes \bfI_n) $.
\end{lemma}

\begin{proof}
Using $\vect(\bfP\bfQ\bfR) = (\bfR^\top \otimes \bfP) \vect(\bfQ)$, we can rewrite $F(\bfx)\ctrlaff$ as
\begin{align}
F(\bfx)\ctrlaff = \vect(\bfI_n F(\bfx)\ctrlaff) = (\ctrlaff^\top \otimes \bfI_n) \vect(F(\bfx)).
\end{align}
Thus, the variance of $F(\bfx)\ctrlaff$ can be expressed in terms of the kernel $\bfSigma \knl_0(\bfx, \bfx')$:
\begin{align}
    F(\bfx)\ctrlaff \!\sim \!\calN\Bigl(
    \bfM_0(\bfx)\ctrlaff, 
    (\ctrlaff^\top \!\!\otimes \bfI_n) \bfSigma \knl_0(\bfx, \bfx') (\ctrlaff \otimes \bfI_n) \Bigr).
\end{align}
The training data is generated from $\dot{\bfx} = F(\bfx)\ctrlaff + \bfw$ with $\bfw \sim \calN(\boldsymbol{0}_n, \bfS)$ and is jointly Gaussian:
\begin{align}
&\vect(\StDtDat_{1:k}) \sim \calN \bigl(\vect(\bcalM_{1:k}\CtDat) , \ubcalK_{1:k}^{1:k} \bigr),
\end{align}
%
The covariance for a single point is given by,
\begin{align}
    &\cov(\vect(F(\bfx)), F(\bfx')\ctrlaff) 
    \\
    &\scaleMathLine{\quad=\cov(\vect(F(\bfx), (\ctrlaff^\top \otimes \bfI_n) \vect(F(\bfx')))
    = \bfSigma \knl_0(\bfx, \bfx') (\ctrlaff \otimes \bfI_n)}.\notag
\end{align}
A similar expression for covariance can be obtained between the training data and the test data, $\cov(\vect(F(\bfx)), \StDtDat_{1:k}) = \ubcalK_{1:k}(\bfx)$.
%
We can now write the joint distribution between the training and test data,
%
\begin{equation*}
\scaleMathLine{%
\begin{bmatrix}
\vect(\StDtDat_{1:k})
\\
\vect(F(\bfx))
\end{bmatrix} 
    \sim \calN\left( 
\begin{bmatrix} 
    \vect(\bcalM_{1:k}\CtDat) \\ 
    \bfM_0(\bfx)
\end{bmatrix}, 
\begin{bmatrix}
    \ubcalK_{1:k}^{1:k} & \ubcalK^\top_{1:k}(\bfx') 
    \\
    \ubcalK_{1:k}(\bfx) & \bfSigma \knl_0(\bfx, \bfx')
\end{bmatrix}
\right).}
\end{equation*}
Applying a Schur complement, we get the posterior distribution in \eqref{eq:coregionalization-posterior}.
\end{proof}

%% file: tex/proof-MVGP-ext-lederer2019uniform.tex
\Revised{}{
\section{Proof of Lemma~\ref{lem:LipschitzGP}} 
\label{sec:proof-MVGP-lederer2019uniform}
\begin{proof}
The bound on each element of the Jacobian of $\bff(\bfx)$ is due to~\cite[Lemma B.2]{lederer2019uniform} with probability at least $1-\frac{\delta_L}{n^2}$,
\begin{align}
    |\bff_i(\bfx') - \bff_i(\bfx)| \le L_{i, \partial j}|\bfx_j' - \bfx_j|\quad \forall i,j \in \{1, \dots, n\}.
    \label{eq:lemmaB2lederer-discrete}
\end{align}
Denote the event in \eqref{eq:lemmaB2lederer-discrete} by $\calE_{i,\partial j}$ so that $\bbP(\calE_{i,\partial j}) \ge 1-\frac{\delta_L}{n^n}$. A lower bound on the probability of intersection of all events can be computed using a union bound:
\begin{align}
    &\bbP(\cap_{i=1}^n \cap_{j=1}^n\calE_{i,\partial j}) = 1 - \bbP(\cup_{i=1}^n  \cup_{j=1}^n \calE_{i,\partial j}^\complement) 
    \notag\\
    &\quad\ge 1 - \sum_{i=1}^n\sum_{j=1}^n\bbP( \calE_{i,\partial j}^\complement)
    = 1 - \sum_{i=1}^n \sum_{j=1}^n(1- \bbP(\calE_{i,\partial j}))
    \notag\\
    &\quad\ge 1-\delta_L.
\end{align}
Adding \eqref{eq:lemmaB2lederer-discrete} for all $j \in \{1, \dots, n\}$ and using the Cauchy-Schwarz inequality, we get for each $i \in \{1, \dots, n\}$:
\begin{equation*}
n|\bff_i(\bfx') - \bff_i(\bfx)| \le \sum_{j=1}^n L_{i, \partial j} |\bfx'_j - \bfx_j| \le \sqrt{\sum_{j=1}^n L_{i, \partial j}^2} \|\bfx' - \bfx\|_2.
\end{equation*}
%
Squaring and adding all inequalities for $i \in \{1, \dots, n\}$, we get:
\begin{align}
 \|\bff(\bfx') - \bff(\bfx)\|_2
 &\le \sqrt{\frac{1}{n^2}\sum_{i=1}^n {\sum_{j=1}^n L_{i, \partial j}^2}} \|\bfx' - \bfx\|_2. \qedhere
\end{align}
\end{proof}
}

%% file: tex/cbc-r-mean-and-variance-affine-and-quadratic.tex

\section{Proof of Lemma~\ref{prob:VarUquadrac24562}}
\label{proof:cbc-r-mean-affine-var-quadratic}

Recall the definition of $\mbox{CBC}^{(r)}(\bfx,\bfu)$:
\begin{align*}
    &\mbox{CBC}^{(r)}(\bfx,\bfu) = 
     \Lie_f^{(r)}h(\bfx)
     + \Lie_g \Lie_f^{(r-1)} h(\bfx) \bfu
    + \bfk_\alpha^\top \eta(\bfx)
    \notag\\
     &= \Lie_f [\Lie_f^{(r-1)}h(\bfx)]
     + \Lie_g [\Lie_f^{(r-1)} h(\bfx)] \bfu
    + \bfk_\alpha^\top \eta(\bfx)
    \notag\\
     &= \grad_\bfx [ \Lie_f^{(r-1)}h(\bfx) ]^\top \!f(\bfx)
     + \grad_\bfx [\Lie_f^{(r-1)} h(\bfx) ]^\top \!g(\bfx) \bfu
    + \bfk_\alpha^\top \eta(\bfx)
    \notag\\
    &= \grad_\bfx [\Lie_f^{(r-1)}h(\bfx)]^\top \dynAff(\bfx) \ctrlaff
    + \bfk_\alpha^\top \eta(\bfx).
\end{align*}%
Since, $\ctrlaff = [1, \bfu^\top]^\top$, we can rewrite the above as
\begin{align*}
    \mbox{CBC}^{(r)}(\bfx,\bfu) = \left(
    \grad_\bfx \Lie_f^{(r-1)}h(\bfx)^\top \dynAff(\bfx) 
    + 
    \begin{bmatrix}\bfk_\alpha^\top \eta(\bfx)
    \\
    \mathbf{0}_{m }\end{bmatrix}^\top
    \right)\ctrlaff
\end{align*}%
Using linearity of expectation, we see that $\E[\mbox{CBC}^{(r)}(\bfx, \bfu)]$ is linear in $\bfu$:
\begin{equation}
\label{eq:exp-cbc-linear}
\begin{aligned}
    &\E[\mbox{CBC}^{(r)}(\bfx, \bfu)] 
    \\
    &=\underbrace{
    \left(
    \E[\dynAff(\bfx)^\top \grad_\bfx \Lie_f^{(r-1)}h(\bfx) ] 
    + \E\left[
    \begin{bmatrix}\bfk_\alpha^\top \eta(\bfx)
    \\
    \mathbf{0}_{m }\end{bmatrix}
    \right]
    \right)^\top
    }_{\bfe^{(r)}(\bfx)}
    \ctrlaff.    
\end{aligned}
\end{equation}
Also, applying $\Var(\bfx^\top \bfu) = \Var(\bfu^\top \bfx) = \bfu^\top \Var(\bfx)\bfu$ to $\mbox{CBC}^{(r)}(\bfx, \bfu)$, we get:
\begin{align}
\Var&[\mbox{CBC}^{(r)}(\bfx, \bfu)] \label{eq:var-cbc-quad-form}\\
 & = \ctrlaff^\top
    \underbrace{
    \Var\left[
    \grad_\bfx \Lie_f^{(r-1)}h(\bfx)^\top \dynAff(\bfx) 
    + 
    \begin{bmatrix}\bfk_\alpha^\top \eta(\bfx)
    \\
    \mathbf{0}_{m }\end{bmatrix}^\top
    \right]
    }_{\bfV^{(r)}(\bfx) \bfV^{(r)}(\bfx)^\top }
    \ctrlaff. \;\;\qed \notag
\end{align}


\section{Proof of Proposition~\ref{prop:456723o433302}}
\label{proof:cbc-r-quad-prog}

Unlike $\mbox{CBC}$ for relative degree one, the distribution of $\mbox{CBC}^{(r)}$ is not a Gaussian Process for $r \ge 2$. Hence, instead of computing the probability distribution analytically, we use Cantelli's inequality to bound the mean and variance of $\CBCr$. For any scalar $\lambda >0 $, we have
\begin{equation}
\label{Caneeekee44!}
    \scaleMathLine{\mathbb{P}\Bigl(\mbox{CBC}_k^{(r)} \ge \E[\mbox{CBC}_k^{(r)}]-\lambda \mid\;  \bfx_k,\bfu_k\Bigr) \ge
    1 - \frac{\Var[\mbox{CBC}_k^{(r)}]}{\Var[\mbox{CBC}_k^{(r)}]+\lambda^2}}.
\notag
\end{equation}
Since we want the probability to be greater than $\tilde{p}_k$, we ensure its lower bound is greater than $\tilde{p}_k$, i.e.,
\begin{align}
    1 - \frac{\Var[\mbox{CBC}_k^{(r)}]}{\Var[\mbox{CBC}_k^{(r)}]+\lambda^2} \ge \tilde{p}_k.
\end{align}
The terms can be rearranged into,
\begin{equation}
\label{eq:Cantelli}
\begin{aligned}
(1-\tilde{p}_k)(\Var[\mbox{CBC}_k^{(r)}]+\lambda^2) &\ge \Var[\mbox{CBC}_k^{(r)}]\\
\lambda &\ge \sqrt{\frac{\tilde{p}_k}{1-\tilde{p}_k}\Var[\mbox{CBC}_k^{(r)}]}.
\end{aligned}
\end{equation}
For some desired margin, $\E[\mbox{CBC}_k^{(r)}] \ge \zeta$, we can substitute  $\lambda= (\E[\mbox{CBC}_k^{(r)}]-\zeta) > 0$ in \eqref{eq:Cantelli}, so that:
\begin{equation}
\label{eq:cbr-r-ineq}
\begin{aligned}
 \E[\mbox{CBC}_k^{(r)}]-\zeta > \sqrt{\frac{\tilde{p}_k}{1-\tilde{p}_k}}
 \Var^{\frac{1}{2}}[\mbox{CBC}_k^{(r)}] 
 \\
\implies  \mathbb{P}\left(\mbox{CBC}_k^{(r)} \ge \zeta \mid \bfx_k,\bfu_k\right) \ge \tilde{p}_k
\end{aligned}
\end{equation}
Using the constraint in \eqref{progg3456755!!} and the explicit expressions in \eqref{eq:exp-cbc-linear} and \eqref{eq:var-cbc-quad-form} for the mean and variance of $\mbox{CBC}_k^{(r)}$, we get:
\begin{align}
\pi(\bfx_k) &\in \argmin_{\bfu_k
\in \mathcal{U}
}\quad  \|\bfR(\bfx_k)\bfu_k\|_2
\label{eq:temp73}\\
\text{s.t. }&
\bfe_k^{(r)}(\bfx_k)^\top\ubfu_k -\zeta
- c^{(r)}(\tilde{p}_k)\left\| 
\bfV_k^{(r)}(\bfx_k) \ubfu_k
\right\|_2
\ge 0, \notag
\end{align}
where $c^{(r)}(\tilde{p}_k) = \sqrt{\frac{\tilde{p}_k}{1-\tilde{p}_k}}$. To ensure that \eqref{eq:temp73} is a standard SOCP, we convert the objective into a SOCP constraint by introducing an auxilary variable $y$, leading to \eqref{progg!!5544i}. \qed


%% file: tex/mean-and-variance-of-cbf-2.tex
\newcommand{\mean}[1]{\bbE[{#1}]}


\section{Proof of Lemma~\ref{lemma:dot-prod-gps}}
\label{proof:dot-prod-of-gps}

Recall the following result about the mean and cumulants of a quadratic form of Gaussian random vectors.

\begin{lemma}[{\cite[p.~55]{searle1971linear}}]
  \label{lemma:quadratic-form}
  Let $\bfx$ be a Gaussian random vector with mean $\bar{\bfx}$ and covariance matrix $\bfSigma$. Let $\bfLambda$ be a symmetric matrix and consider the random variable $\bfx^\top \bfLambda \bfx$. The mean of $\bfx^\top \bfLambda \bfx$ is
  \[
    \E[\bfx^\top \bfLambda \bfx] = \bar{\bfx}^\top \bfLambda \bar{\bfx} + \tr(\bfLambda \bfSigma ).
  \]
  The $r$th cumulant of $\bfx^\top \bfLambda \bfx$ is:
  \[
  \calK_r(\bfx^\top \bfLambda \bfx) = 2^{r-1}(r-1)![\tr(\bfLambda \bfSigma)^r + r\bar{\bfx}^\top \bfLambda (\bfSigma \bfLambda)^{r-1} \bar{\bfx}].
  \]
  The variance of $\bfx^\top \bfLambda \bfx$ is the second cumulant:
  \[
  \Var[\bfx^\top \bfLambda \bfx] = \calK_2(\bfx^\top \bfLambda \bfx) = 2\tr(\bfLambda \bfSigma)^2 + 4\bar{\bfx}^\top \bfLambda \bfSigma \bfLambda \bar{\bfx}.
  \]
  The covariance between $\bfx$ and $\bfx^\top \bfLambda \bfx$ is:
  \[
  \cov(\bfx, \bfx^\top \bfLambda \bfx) = 2\bfSigma \bfLambda \bar{\bfx}
  \]
\end{lemma}

Returning to the proof of Lemma~\ref{lemma:dot-prod-gps}, consider the three Gaussian random vectors $\bfx$, $\bfy$, and $\bfz$. Note that $\bfx^\top \bfy$ can be written as a quadratic form:
\begin{equation}
\label{eq:quadratic-form}
\bfx^\top \bfy = \frac{1}{2}\begin{bmatrix}\bfx^\top & \bfy^\top & \bfz^\top\end{bmatrix}
\begin{bmatrix} 0 & \bfI & 0 \\ \bfI & 0 & 0 \\ 0 & 0 & 0 \end{bmatrix}
\begin{bmatrix} \bfx \\ \bfy \\ \bfz\end{bmatrix}.
\end{equation}
Applying Lemma~\ref{lemma:quadratic-form} to \eqref{eq:quadratic-form}, shows that:
\begin{align*}
&\E[\bfx^\top \bfy] = \bar{\bfx}^\top \bar{\bfy} + \frac{1}{2}\tr(\cov(\bfx, \bfy) + \cov(\bfy, \bfx))\\
&\Var[\bfx^\top \bfy] = \frac{1}{2}\tr(\cov(\bfx, \bfy)+\cov(\bfy, \bfx))^2 + \bar{\bfy}^\top \Var[\bfx] \bar{\bfy}\\
&\quad+ \bar{\bfx}^\top \Var[\bfy] \bar{\bfx }
+ \bar{\bfy}^\top \cov(\bfx, \bfy) \bar{\bfx}
+ \bar{\bfx}^\top \cov(\bfy, \bfx) \bar{\bfy}\\
&\begin{bmatrix}
    \cov(\bfx, \bfx^\top \bfy)
    \\
    \cov(\bfy, \bfx^\top \bfy)
    \\
    \cov(\bfz, \bfx^\top \bfy)
\end{bmatrix}
  = \cov\left(\begin{bmatrix}
    \bfx \\ \bfy \\ \bfz
  \end{bmatrix},
  \bfx^\top\bfy\right)
  \notag\\
  &=
  2 \begin{bmatrix}
  \Var[\bfx] & \cov(\bfx, \bfy) & \cov(\bfx, \bfz) \\
  \cov(\bfy, \bfx) & \Var[\bfy] & \cov(\bfy, \bfz) \\
  \cov(\bfz, \bfx) & \cov(\bfz, \bfy) & \Var[\bfz]
  \end{bmatrix}
  \frac{1}{2}\begin{bmatrix} \bar{\bfy} \\ \bar{\bfx} \\ \mathbf{0} \end{bmatrix}. \tag*{\qed}
\end{align*}